\documentclass[format=acmsmall, review=false]{acmart}

\usepackage{acm-ec-26}
\usepackage{booktabs} % For formal tables
\usepackage[ruled]{algorithm2e} % For algorithms

\AtBeginDocument{%
  \fancypagestyle{firstpagestyle}{%
    \fancyhf{}%
  }%
}

\makeatletter
\def\@ACM@checkaffil{}%
\makeatother

\SetAlFnt{\small}
\SetAlCapFnt{\small}
\SetAlCapNameFnt{\small}
\SetAlCapHSkip{0pt}
\IncMargin{-\parindent}

\setcitestyle{authoryear}

\usepackage{amsmath}
\usepackage{mathtools}

\usepackage{amsfonts}
\usepackage{graphicx}
\usepackage{tikz}
\usepackage{array,booktabs,multirow}

\usepackage{parskip}
\usepackage{setspace}

\usepackage{cleveref}
\usepackage{enumitem}
\newcommand{\sw}{\text{sw}}
\newcommand{\D}{\mathcal{D}}

\usepackage{xcolor}

\DeclarePairedDelimiter{\ceil}{\lceil}{\rceil}
\DeclarePairedDelimiter{\floor}{\lfloor}{\rfloor}
\def\mc{\ensuremath\mathcal}

\newcommand{\mech}{M}

\DeclareMathOperator{\E}{\mathbb{E}}

\DeclareMathOperator{\rank}{\mathrm{rank}}

\title[Truthful Fair Division under Stochastic Valuations]{Truthful Fair Division under Stochastic Valuations}

\author{Daniel Halpern}
\affiliation{%
  \institution{Google Research}
}
\email{dhalpern@google.com}

\author{Alexandros Psomas}
\affiliation{%
  \institution{Purdue University}
}
\affiliation{%
  \institution{Google Research}
}
\email{apsomas@purdue.edu}

\author{Shirley Zhang}
\affiliation{%
  \institution{Harvard University}
}
\email{szhang2@g.harvard.edu}

\begin{abstract}
We study no-money mechanisms for allocating indivisible items to strategic agents with additive preferences under a stochastic model. In this model, items' values are drawn from an underlying distribution and mechanisms are evaluated with respect to this draw (e.g., in expectation, or with high probability). Motivated by worst-case impossibilities which show that truthfulness severely restricts fairness and efficiency, we ask whether truthful mechanisms continue to perform poorly on random instances.

We first focus on dominant-strategy incentive compatible (DSIC) mechanisms. For two agents, we obtain a tight picture. Specifically, we show that there exists a distribution under which no DSIC mechanism achieves an expected welfare approximation better than $\frac{2+\sqrt{2}}{4}\approx 0.854$, and we give a DSIC mechanism that matches this bound for all distributions simultaneously. We further show that, for every distribution, there exists a DSIC mechanism that is envy-free with high probability and obtains the same welfare. A key ingredient is a new, tight connection between welfare guarantees of a family of DSIC, no-money mechanisms and i.i.d.\ prophet inequalities. This connection allows us to generalize to $n$ agents; in particular, we obtain a DSIC mechanism that achieves a $\approx 0.745$ approximation to welfare, and another DSIC mechanism achieving a $1/2$-approximation welfare that is envy-free with high probability.

We then turn to Bayesian incentive compatibility (BIC). Under i.i.d.\ valuations, we show that BIC comes at essentially no cost: we design a prior-independent BIC mechanism that achieves a $(1-\varepsilon)$-approximation to the optimal welfare, while being envy-free with high probability. Under independent but non-identical priors, we obtain BIC mechanisms that are $(1-\varepsilon)$-approximately Pareto efficient and envy-free with high probability.
\end{abstract}

\begin{document}

% Title page for title and abstract only.
\begin{titlepage}

\maketitle
\end{titlepage}

\section{Introduction}

We study the allocation of indivisible items among strategic agents with additive preferences in the absence of monetary transfers. In this classic setting, the gold standard is mechanisms that are simultaneously fair, economically efficient (e.g., Pareto optimal), and strategically robust (e.g., truthful reporting is a dominant strategy for every agent).

While the literature has uncovered a plethora of natural rules that are both fair and efficient, the same cannot be said for combinations with strategic robustness. Impossibility results paint a bleak picture.
For instance, deterministic truthful mechanisms cannot guarantee envy-freeness up to one good (EF1)~\cite{amanatidis2017truthful}, a canonical fairness notion for indivisible goods~\cite{lipton2004approximately}. This can be extended to $\text{EF}^{+u}_{-v}$ for any fixed choice of $u$ and $v$~\cite{bu2024truthful}.\footnote{An allocation satisfies $\text{EF}^{+u}_{-v}$ if no agent envies another agent after adding $u$ items to their bundle and removing $v$ items from the other agent's bundle.}  Moreover, the only mechanism that is both Pareto optimal and truthful is a serial dictatorship~\cite{schummer1996strategy}, arguably the least fair mechanism imaginable. In short, mechanisms that are simultaneously fair, efficient, and truthful do not exist. As a result, much of the fair division literature has largely set aside issues of incentives and instead focused purely on fairness and efficiency.

However, a core tenet in economics is that agents \emph{are} strategic. If agents can easily manipulate the outcome, how can we trust the promised fairness guarantees of the resulting allocations? This, on the surface, seems to create a fundamental dilemma for the field.

We posit that all hope may not be lost. A closer look at the aforementioned negative results reveals a commonality: they all hold in the \emph{worst case}, relying on carefully constructed pathological instances.
To understand the performance beyond the worst case,  much of the recent fair division literature has adopted distributional assumptions. For example, in asymptotic fair division it is well understood that fairness guarantees that are unattainable in the worst case (e.g., envy-freeness) are often achievable with high probability on random instances. What is far less understood is whether the same phenomenon holds for incentives. Do truthful mechanisms --- that can't meaningfully balance fairness and efficiency under worst-case analysis --- perform just as poorly on random instances?

In this paper, we study truthful no-money mechanisms under stochastic valuations. Our main question is whether such mechanisms can provide meaningful fairness and efficiency guarantees with high probability. Put differently, is it possible that by going beyond worst-case analysis we can actually achieve the fairness--efficiency--truthfulness trifecta on most fair division instances?

\subsection{Our Contribution}

A subtlety arises when considering incentives under stochastic valuations, which is that there is no longer a single canonical notion of truthfulness. The definition depends on whether an agent considers the consequences of their action ex-post or in expectation over the other agents' values. In this paper, we consider both \emph{Dominant Strategy Incentive Compatibility} (DSIC) and \emph{Bayesian Incentive Compatibility} (BIC) to more deeply understand the trade-offs imposed by these constraints.

We begin by focusing on DSIC mechanisms. In our stochastic model, agents' valuations are drawn from a known distribution $\D$, and the mechanism is evaluated with respect to the random draw of valuations (e.g., in expectation or with high probability).

We start with the two-agent case.
In this setting, we can leverage a known characterization of DSIC no-money mechanisms and restrict attention to the class of \emph{picking--exchange} mechanisms~\cite{amanatidis2017truthful}. We give a complete picture.
First, we show that there exists a distribution $\D$ such that no DSIC mechanism can, in expectation, achieve an $\alpha$-approximation to the optimal welfare for any $\alpha > \frac{2+\sqrt{2}}{4}\approx 0.854$ (\Cref{thm:neg-2}). Second, we show that there exists a mechanism that achieves, in expectation, a $\frac{2+\sqrt{2}}{4}$-approximation to the optimal welfare, for all underlying distributions \emph{simultaneously} (\Cref{thm:pos-welf-2}). Furthermore, we identify a pair of mechanisms, such that, given a distribution $\D$, at least one of the two mechanisms will achieve a $\frac{2+\sqrt{2}}{4}$-approximation to the optimal welfare, as well as be envy-free with high probability (\Cref{thm: main positive result for two agents}).

A notable aspect of our positive results is that they reveal a tight and, to the best of our knowledge, previously unnoticed connection to the rich literature on prophet inequalities. Concretely, in~\Cref{thm:prophet-connection} we show that the welfare approximation ratio of a natural class of \emph{picking-sequence} mechanisms (which are DSIC and do not use monetary transfers) coincides with the optimal guarantees in the i.i.d. prophet inequality problem. This correspondence explains the appearance of the constant $\frac{2+\sqrt{2}}{4}$ (which is the best possible ratio in a specific instance of the i.i.d. prophet inequality problem) in the two-agent case, and furthermore provides a principled route for generalizing our positive results to $n > 2$ agents. Specifically, for $n > 2$ agents, we prove that there exists a DSIC mechanism that achieves an $\approx 0.745$ approximation to welfare (\Cref{thm:correa-pos-n}) and that a simple DSIC mechanism gives a $1/2$ approximation to the optimal expected welfare, as well as is envy-free with high probability (\Cref{thm: n agents dsic 2 approx}).

Finally, we turn to BIC mechanisms.
Under i.i.d.\ stochastic valuations, it is known that it is possible to simultaneously achieve optimal welfare and envy-freeness with high probability~\cite{dickerson2014computational}. We prove that BIC comes at essentially no cost (!); specifically, we give a BIC mechanism that guarantees a (1 - $\varepsilon$)-approximation to the optimal welfare and is envy-free with high probability (\Cref{thm: bic iid theorem}). Our mechanism requires no knowledge of the underlying distribution (it just assumes that an underlying distribution exists).
This result can also be extended to independent but non-identical distributions per agent. Under the same conditions as~\cite{baienvy2022} (see related work), when the underlying distributions are known, we can get a mechanism that is BIC, (1 - $\varepsilon$)-approximately Pareto efficient and envy-free with high probability.

\subsection{Related Work}

\paragraph{Incentives in Fair Division}

Numerous impossibility results in fair division say that ``fairness'' and ``truthfulness'' are incompatible for various fairness notions (where ``truthfulness" is used as a synonym to DSIC)~\cite{bu2024truthful,amanatidis2017truthful,amanatidis2016truthful,markakis2011worst,lipton2004approximately}. Furthermore, (Pareto) efficiency and truthfulness are known to imply dictatorships~\cite{schummer1996strategy,papai2000strategyproof,papai2001strategyproof}. In fact, a truthful, non-bossy, and neutral mechanism must be a serial-quota mechanism~\cite{babaioff2025truthful}.
Therefore, much of the work in truthful fair division attempts to achieve fairness, efficiency, and incentive compatibility by restricting the family of instances considered, e.g., dichotomous valuations~\cite{bu2024truthful,halpern2020fair,barman2022truthful,amanatidis2021maximum,babaioff2021fair}, Leontief valuations~\cite{ghodsi2011dominant,friedman2014strategyproof,parkes2015beyond}, or a small number (e.g., $n=2$) agents~\cite{bu2024truthful}.
Some works use ``money burning,'' i.e., leave some resources unallocated, to recover incentive compatibility in addition to fairness and efficiency~\cite{cole2013mechanism,friedman2019fair,abebe2020truthful}. 

In all aforementioned works, ``truthfulness'' is a synonym for ``dominant strategy incentive compatibility (DSIC).'' Some recent works in fair division explore whether the aforementioned impossibility results can be bypassed by relaxing DSIC to truthful in expectation~\cite{mossel2010truthful}, Bayesian incentive compatibility~\cite{gkatzelis2024getting}, not obviously manipulable~\cite{psomas2022fair,ortega2022obvious}, or partial strategyproofness~\cite{mennle2021partial}.~\cite{hartman2025s} explores how manipulable different mechanisms are. Finally, another recent, related line of work explores whether the Nash equilibria of non-truthful mechanisms (e.g., the round robin procedure) satisfy fairness guarantees~\cite{amanatidis2024allocating,amanatidis2023round}.

\paragraph{Asymptotic Fair Division}

\citet{dickerson2014computational} initiate the study of fair division under stochastic preferences, and proved that the welfare maximizing allocation is envy-free with high probability when the number of goods $m \in \Omega(n \log n)$ (where $n$ is the number of agents) and items' values are drawn i.i.d. from a fixed distribution with bounded PDF. \citet{manurangsi2021closing} improve this result, and show that an envy-free allocation exists with high probability when $m \in \Omega(n \log n/ \log\log n)$. \citet{baienvy2022} study the existence of envy-free allocations for the case of independent but non-identically distributed agents' values. They show that under mild restrictions on the distribution (in particular, the distribution PDF values are both lower and upper bounded on the non-zero interval), envy-free and Pareto optimal allocations exist with high probability. Furthermore, such an allocation is achieved by maximizing \emph{weighted} welfare, where the weights depend on the agents' specific distributions.
\citet{benade2024existence} study independent agents with non-additive stochastic valuations, and \citet{benade2024fair} study non-independent agents in an online setting. We note that all aforementioned works treat the distributions as constants; \citet{halpern2025online} show that envy-free allocations exist with high probability even without this assumption, and even in an online setting.

Closer to our work, and to the best of our knowledge the only other work that studies incentives in the asymptotic regime, \citet{garg2025designing} give a truthful in expectation mechanism that outputs envy-free allocations with high probability, even in a model where agents' valuations are correlated. In contrast, in this work, we give DSIC mechanisms (a stronger notion of incentive compatibility) that are envy-free with high probability, as well as approximately maximize welfare; alas, we require agents to be independent.

We will also make use of some results from i.i.d.\@ prophet inequalities. We will introduce the necessary preliminaries and discuss related work when they are used in \Cref{sec:n-agents}.

\section{Model}

A set of $m$ indivisible items $\mathcal{M} = \{1,2, ..., m\}$ must be allocated among a set of $n$ agents $\mathcal{N} = [n]$ with additive preferences. An \emph{allocation} $A = \{A_1, A_2,...,A_n\}$ partitions the $m$ items in $\mathcal{M}$ among the $n$ agents, where the set of items allocated to agent $i\in [n]$ is $A_i$ (we refer to $A_i$ as agent $i$'s allocation or bundle). Note that each item $j \in \mathcal{M}$ must be assigned to exactly one agent in $\mathcal{N}$. Each agent $i \in \mathcal{N}$ has a valuation vector $v_i \in \mathbb{R}^m_{\geq 0}$, where $v_{i,j}$ is agent $i$'s value for item $j$. The utility of agent $i$ for an allocation $A$ to be $u_i(A) = \sum_{j \in A_i} v_{i,j}$. 

We assume that agents' values are drawn i.i.d. from a distribution $\D$, i.e. $v_{i,j} \sim \D$ for all $i \in \mathcal{N}$ and $j \in \mathcal{M}$. We make some minimal assumptions on $\D$, namely, that it is nonnegative, bounded, has nonzero variance, and continuous.\footnote{Strictly speaking, continuous distributions are not necessary for any of our results; we simply assume it for ease of exposition. This ensures quantiles for item values are well-behaved. For discrete distributions, care must be taken to handle tie-breaking.} We call a distribution satisfying these conditions \emph{valid}. Note that none of our results will depend on the rescaling of $\D$; thus, it is without loss of generality to assume that $\D$ is supported on $[0, 1]$.  

While values $v_i$ are private and known only to agent $i$, agent $i$ can strategically report their values as a vector of bids $b_i$, where $b_{i,j}$ represents agent $i$'s bid for item $j$. Note that the agents are strategic and therefore may misreport their bids, so $b_i$ is not necessarily equal to $v_i$. A mechanism is a function  $\mech(\cdot)$ that maps a profile of bids $\textbf{b} = (b_1,...,b_n)$  to an allocation $\mech(\textbf{b})$.

A mechanism is \emph{incentive compatible} if it is robust to strategic agents. In this paper, we focus on two formalizations of incentive compatibility. First, a mechanism is \emph{dominant strategy incentive compatible (DSIC)} if truthfully bidding their private values is a dominant strategy for every agent. Formally, a mechanism  $\mech(\cdot)$ is DSIC if, for every agent $i \in \mathcal{N}$, every possible bid $b_i$, and every possible bids of other agents $b_{-i}$, 
\[
    u_i(\mech(v_i, b_{-i})) \ge u_i(\mech(b_i, b_{-i})).
\]
Second, a mechanism is \emph{Bayesian incentive compatible (BIC)} if every agent truthfully bidding their private values is a Bayesian Nash equilibrium. This means that no agent can achieve higher expected utility by bidding anything other than their true valuations, where the expectation is taken over the random values of all other agents. Formally, a mechanism $\mech(\cdot)$ is BIC if for every possible bid $b_i$, 
\[
    \E_{v_{-i} \sim \D^{n-1}}[u_i(\mech(v_i, v_{-i}))] \ge \E_{v_{-i} \sim \D^{n-1}}[u_i(\mech(b_i, v_{-i}))].
\]

An allocation $A$ is \emph{envy-free} (EF) if, for every pair of agents $i_1,i_2 \in \mathcal{N}$, agent $i_1$ has a higher utility for their own bundle than for agent $i_2$'s bundle. Formally, for all $i,j \in \mathcal{N}$, $u_{i_1}(A_{i_1}) \geq u_{i_1}(A_{i_2})$. 

The social welfare of an allocation $A$ is the total utility across all agents, i.e. $\sw^{\mathbf{u}}(A) = \sum_{i \in \mathcal{N}} u_i(A)$.
We say a mechanism $\mech$ achieves an $\alpha$ approximation to welfare if \[\E_{\mathbf{u}}[\sw^\mathbf{u}(\mech(\mathbf{v}))] \ge \alpha \cdot \E_{\mathbf{u}}[\max_{A} \sw^\mathbf{u}(A)].\]
We will say the mechanism achieves an $\alpha - o(1)$ approximation to welfare if the approximation approaches $\alpha$ as $m$ grows large. Furthermore, we will say a mechanism is envy-free with high probability if the probability it outputs an envy-free allocation approaches $1$ as $m$ grows large.

\section{DSIC Mechanisms}

\subsection{Two Agents}

We begin by studying DSIC mechanisms for the case of $n = 2$ agents. In the context of stochastic values, one might hope that the relaxation from worst-case analysis would allow for mechanisms that approximate the optimal welfare arbitrarily well. However, as we show next, this is not the case; there is a constant ceiling on the welfare approximation achievable by any DSIC mechanism.

\begin{theorem}\label{thm:neg-2}
    When $n = 2$, for all $\varepsilon > 0$, there exists a valid distribution $\D$ such that no DSIC mechanism achieves a $\frac{2 + \sqrt{2}}{4} + \varepsilon$ approximation to welfare. 
\end{theorem}

The proof of this result builds on a structural characterization of DSIC mechanisms for two agents established by \citet{amanatidis2017truthful}. At a high level, every DSIC mechanism for two agents can be decomposed into two fundamental types of mechanisms: \emph{picking} mechanisms and \emph{exchange} mechanisms.
A picking mechanism is defined by a set of allowable subsets $\mathcal{O} \subseteq \mathcal{P}(\mathcal{M})$ and a designated agent $i$. Given the reported bids, agent $i$ receives the subset $A \in \mathcal{O}$ that maximizes their perceived value, while the other agent receives the remainder $\mathcal{M} \setminus A$. This is essentially a ``menu'' approach where one agent chooses their favorite bundle from a restricted list.

An exchange mechanism, by contrast, relies on a fixed partition of items $\mathcal{M}$ into sets $S$ and $T$. We can think of agent 1 as being ``endowed" with $S$ and agent 2 with $T$. A trade is considered favorable for an agent if they strictly prefer the other agent's endowment to their own; it is considered unfavorable if they strictly prefer their own endowment to the other agent's. An exchange mechanism outputs $(T, S)$ if the trade is favorable for both agents and $(S, T)$ if the trade is unfavorable for either. (In the cases of indifference, the mechanism must still select and output one of $(S, T)$ or $(T, S)$.)

\citet{amanatidis2017truthful} prove that all DSIC mechanisms for two agents are \emph{picking-exchange} mechanisms: the items are partitioned into subsets $\mathcal{M} = M_1 \sqcup \cdots \sqcup M_\ell$ and each subset $M_k$ is allocated according to a fixed picking or exchange mechanism. 
Our proof for \Cref{thm:neg-2} proceeds by providing a single distribution for which neither of these building blocks can perfectly capture the social welfare.  The distribution we consider takes value $0$ or $1$ most of the time, and with small probability takes a (very) high value $H$. For this distribution, we then show that no picking or exchange mechanism can achieve higher than a $\frac{2+\sqrt{2}}{4} + \varepsilon$ approximation to welfare. 

Intuitively, a picking mechanism fails to achieve a high approximation to welfare because, in a picking mechanism, the allocation of items is determined solely by agent 1. Agent 1 does not know which items agent 2 has a high value for, and therefore, in expectation, agent 1 will pick some items where they have value 1 for the item but agent 2 has value $H$ for the item. This leads to an avoidable loss in expected social welfare for any picking mechanism, which we show is at least a multiplicative factor of $\frac{2+\sqrt{2}}{4} + \varepsilon$.

Bounding the expected welfare of all exchange mechanisms requires a more involved argument. Specifically, we bound the expected welfare contribution of all non-high-value items and all high-value items separately. For the former, we observe that the exchange component $(S,T)$ must output either the allocation $(S,T)$ or $(T,S)$; therefore, we can bound the contribution of all non-high-value items by bounding the expected maximum welfare of all non-high-value items in $(S,T)$ or $(T,S)$. The bound for high-value items is technically more involved. To bound the expected welfare contribution of all high-value items, we first condition on the total number of high-value items and then bound the welfare contribution based on correctly and incorrectly assigned high-value items. Finally, we can combine the expected welfare contribution of the low-value items and the high-value items to show that any exchange mechanism is at best a $\frac{2+\sqrt{2}}{4} + \varepsilon$ approximation.

\begin{proof}[Proof of \Cref{thm:neg-2}] Fix $\varepsilon > 0$. We define a family of distributions parameterized by positive constants $p$ and $\delta$. Taking $p$ and $\delta$ sufficiently small (in terms of $\varepsilon$) will give the desired result.

To draw a value $v$, we first sample a ``base" value:
\[
v_{\text{base}} =
\begin{cases}
    \frac{1 - q}{p} & \text{with probability } p \\
    1 & \text{with probability } q \\
    0 & \text{with probability } 1 - q - p
\end{cases}
\]
where $q = 2 - \sqrt{2} \approx 0.59$. To this base value, we add uniform noise scaled by $\delta$, so we have
$v = v_{\text{base}} + \text{Unif}[0, \delta]$. The noise term's only purpose is to ensure that ties
in agents' valuations for distinct subsets of goods occur with probability 0. We will refer to the above distribution (for sampling a value $v$) as $\mathcal{D}$ throughout this proof. 

The expected value of a single draw from $\mathcal{D}$ is:
\begin{equation}\label{eq:exp_value_single_draw}
\E[v] = p \cdot \frac{1-q}{p} + q \cdot 1 + (1-q-p) \cdot 0 + \frac{\delta}{2} = (1-q)+q = 1 + \frac{\delta}{2}.
\end{equation}

\textbf{The Welfare Benchmark.} The expected optimal social welfare per item is the expected maximum
of two i.i.d. draws $v_1$ and $v_2$ from $\mathcal{D}$. The probability that at least one draw is a ``high value"
(i.e., $v_{\text{base}} = \frac{1-q}{p}$) is $1-(1-p)^2 = 2p - p^2$ and the probability that at least one draw is $1$ and no draw is a high value is $q^2 + 2q(1-q-p)$.

The expected value of the max of two draws $v_1$ and $v_2$ is therefore:

\begin{align*}
    \E[\max(v_1,v_2)] &= \frac{1-q}{p}\left(2p - p^2 \right) + 1 \cdot (q^2 + 2q(1-q-p))+ \delta/2 \\
    &= (1-q)(2-p) + 2q - q^2 - 2pq + \delta/2\\
    &= 2 - q^2 - qp - p + \delta/2 \\
    &= 4\sqrt{2} - 4 - qp - p + \delta/2.
\end{align*}

The rest of the proof will focus on showing that no DSIC mechanism can achieve expected welfare greater than 
 $\sqrt{2} + \psi(\delta,p)$ per item, for some function $\psi$ that goes to zero as $\delta$ and $p$ go to zero.
Since we already showed that the maximum expected welfare is $4\sqrt{2}-4  - qp - p + \delta/2$ per item, this establishes the desired approximation ratio upper bound of

\begin{equation}\label{eq:overall_to_show}
\frac{\sqrt{2} + \psi(\delta,p)}{4\sqrt{2}-4 - qp - p + \delta/2} \le \frac{\sqrt{2}}{4\sqrt{2} - 4} + \varepsilon =  \frac{2+\sqrt{2}}{4} + \varepsilon,
\end{equation}
for sufficiently small $\delta$ and $p$.

\textbf{Mechanism Decomposition.} From the results in \cite{amanatidis2017truthful}, any DSIC mechanism for two agents
partitions the set of items $M$ into components and runs a distinct sub-mechanism on each component: $M = P_1 \sqcup P_2 \sqcup E_1
\sqcup \cdots \sqcup E_{k}$. The first component is \emph{Picking Mechanisms} (for allocating subsets $P_1$ and  $P_2$), where an
agent is allocated their most preferred bundle from a collection of subsets (subsets of $P_1$ and $P_2$, respectively), and the other agent receives the remaining goods. Note that two picking mechanisms suffice; the first one, where agent 1 picks first, allocates the subset $P_1$, and the second one, where agent 2 picks first, allocates the subset $P_2$.
The second component is  \emph{Exchange Mechanisms}
(for allocating subsets $E_1, \dots, E_k$), where items are endowed to the agents, and traded if and only if both agents prefer the other's bundle. We will prove that for any component corresponding to a set of items $C$ from this partition (with $|C| = s$), the expected welfare generated by these items is at
most 
\begin{equation}\label{eq:to_show}
\left(\sqrt{2} + \psi(\delta,p)\right) \cdot |C| = \left(\sqrt{2} + \psi(\delta,p)\right) s,
\end{equation}
for some function $\psi(\delta,p)$ that goes to zero as $\delta$ and $p$ go to zero.

\textbf{Analysis of Picking Mechanisms.} Without loss of generality, consider $P_1$ with $|P_1| = s$. Although the fraction of items that agent $1$ picks can be arbitrary, the key property we will use is that the allocation rule for $P_1$ depends only on agent 1's values. Condition on agent $1$'s values ($v_{1, j}$ for $j \in P_1$) but assume we have not yet sampled agent $2$'s values.

Suppose agent 1 is allocated a set $S \subseteq P_1$ and agent $2$ gets $P_1 \setminus S$. Let $H = \{j \in P_1 \mid v_{1, j} \geq \frac{1 - q}{p} \}$ be the set of items agent $1$ has high (base) value for in $P_1$, and let $\mathcal{E}_H$ be the event that $H$ is the set of items for which agent $1$ has high  (base) value for. For any $j \in P_1$, if $j \in H$ and agent 1 is given item $j$, then this item contributes at least $\frac{1 - q}{p}$, and at most $\frac{1 - q}{p} + \delta$, to the welfare. For any $j \in P_1$ such that $j \not\in H$, if agent 1 is given item $j$, then this item contributes at most $1 + \delta$ to the welfare. For any item $j \in P_1$, if we give item $j$ to agent 2, then the expected contribution of item $j$ to the welfare, conditioned on $\mathcal{E}_H$, is $\E[v_{2,j} \mid \mathcal{E}_H] = 1 + \frac{\delta}{2}$ by Equation \eqref{eq:exp_value_single_draw}. Therefore, the conditional expected welfare of the component $P_1$ is at most 
\[
    \E[\text{welfare of $P_1$} \mid \mathcal{E}_H] \le |H| \cdot \left( \frac{1-q}{p} + \delta \right) + (|P_1 \setminus H|) \cdot \left( 1+ \delta \right) \leq s + |H|\left(\frac{1-q-p}{p}\right) + \delta s.
\]
Taking an expectation over $H$ gives that the expected welfare of the component $P_1$ is at most
\begin{align*}
    \E[\text{Welfare of $P_1$}] &= \E\left[ \E[\text{welfare of $P_1$} \mid \mathcal{E}_H]\right] \\
    &=    \E\left[ s + |H|\left(\frac{1-q-p}{p}\right)  + \delta s\right] \\
    &= s + \left(\frac{1-q-p}{p}\right) \E[|H|] + \delta s \\
    &= s + \left(\frac{1-q-p}{p}\right) ps + \delta s \\
    &= s(2 - q - p)  + \delta s\\
    &\le (2-q)s  + \delta s
    = \sqrt{2} s + \delta s = (\sqrt{2} + \delta)s,
\end{align*}
as desired, where we used that $\E[|H|] = ps$. The analysis for $P_2$ is identical by symmetry. Therefore, we have shown Equation \eqref{eq:to_show} for picking mechanisms.

\textbf{Analysis of Exchange Mechanisms.} Consider an exchange component $E$ where agent 1 is endowed
with $\ell$ items and agent 2 with $k$ items. Let $s = |E| = \ell+k$. The following argument will work for all $s \ge 3$. At the end of this proof, we separately handle the final nondegenerate case of $\ell = k = 1$. (Note that it is without loss of generality that $\ell \geq 1$ and $k \geq 1$, as otherwise these items are always allocated to a single agent and trade never occurs; this can therefore be represented as part of the picking mechanism.)

We bound the welfare by analyzing the contributions from non-high-value and high-value items separately, where a high-value item is an item with base value $\frac{1-q}{p}$.

\textit{{Contribution from Non-High-Value Items.}} For any item, if an agent's value for that item is not ``high,'' then the item's value is drawn from a $\text{Bern}(\frac{q}{1-p}) + \text{Unif}(0, \delta)$ distribution. Let $W$ be the welfare of non-high-value items in the endowed allocation. Let $W'$ be the welfare of the non-high-value items in the ``flipped'' allocation, where the $\ell$ items given to agent 1 are given to agent 2 instead, and the $k$ items given to agent 2 are given to agent 1. We can decompose $W$ as $W = W_{base} + Z$, where $W_{base}$ is the sum of the Bernoulli contributions, and $Z$ is the sum of the noise terms from $\text{Unif}(0, \delta)$. Note that $W_{base}$ follows a binomial distribution $W_{base} \sim \text{Bin}(s, \frac{q}{1-p}$, and the noise term is bounded by $0 \le Z \le s \delta$ with probability $1$. Similarly, $W' = W'_{base} + Z'$ where $W'_{base}$ is an independent draw from the same binomial distribution and $0 \le Z' \le s \delta$. 

The expected welfare of component $E$ from non-high-values can be upper bounded by $\E[\max(W, W')]$. Using the deterministic bound on the noise, we have:$$\E[\max(W, W')] \le \E[\max(W_{base} + s\delta, W'_{base} + s\delta)] = \E[\max(W_{base}, W'_{base})] + s\delta.$$

Since $W_{base}$ and $W_{base}'$ are independent and identically distributed from $\text{Bin}\left(s, \frac{q}{1-p}\right)$,

we have:
\begin{align*}
\mathbb{E}[\max(W_{base}, W'_{base})] &= \mathbb{E}[W_{base}] + \frac{1}{2}\mathbb{E}[|W_{base} - W'_{base}|] \\
&\le \mathbb{E}[W_{base}] + \frac{1}{2}\sqrt{\text{Var}(W_{base} - W'_{base})} \\
&= \mathbb{E}[W_{base}] + \sqrt{\frac{\text{Var}(W_{base})}{2}}.
\end{align*}

where the second transition follows from Jensen's inequality, and the third follows because the variance of the sum of independent random variables is the sum of their variances. Substituting the mean $\mathbb{E}[W_{base}] = s\frac{q}{1-p}$ and variance $\text{Var}(W_{base}) = s\frac{q}{1-p}(1 - \frac{q}{1-p})$, and combining with the noise error term $s\delta$, we obtain the upper bound:\[sq + \sqrt{\frac{sq(1-q)}{2}} + s \delta .\]

\textit{Contribution from High-Value Items.} Let $h$ be the total count of high-value draws in the endowed and non-endowed items by both agents, i.e., $h = |\{(i, j) \mid v_{i, j} \geq \frac{1 - q}{p} \}|$. Note that
$\E[h] = 2sp$. We will show the following lemma, the proof of which can be found in Appendix \ref{sec:high_welfare_conf_on_h}.

\begin{lemma}\label{lemma:high_welfare_conf_on_h}
    Conditioned on any value of $h$, the expected welfare of high-value items in this component is at most
    \[
        \frac{3}{4}  \, \frac{1-q}{p} \, h + \delta h.
    \]
\end{lemma}

Using that $\E[h] = 2sp$, this gives that the expected welfare contribution of high-value items is at most
\[
    \E\left[\frac{3}{4}  \, \frac{1-q}{p} \, h + \delta h \right] = \left(\frac{3}{4} \,  \frac{1-q}{p} + \delta \right) (2sp)
    = \frac{3}{2}(1-q)s + 2 p \delta s.
\]

\paragraph*{Combining the Bounds.} The total expected welfare is bounded by the sum of the contributions from non-high-value and high-value items, which from above is upper bounded by:

\[
\left( sq + \sqrt{\frac{sq(1-q)}{2}} + \delta s \right) + \left( \frac{3}{2}(1-q)s + 2 p \delta s \right) = \left( \frac{3-q}{2} \right)s + \sqrt{\frac{sq(1-q)}{2}} + s \cdot \psi(\delta, p),
\]
where $\psi$ goes to $0$ as $\delta$ and $p$ go to $0$. 
Our goal is to show that for $s \ge 3$, this is at most $s(\sqrt{2} +  \psi(\delta, p))$. In other words, we will show that:
\[
\left( \frac{3-q}{2} \right)s + \sqrt{\frac{sq(1-q)}{2}} \le \sqrt{2} s.
\]
Rearranging and plugging in $q = 2-\sqrt{2}$, it suffices to show that
\[
\sqrt{\frac{s(3\sqrt{2}-4)}{2}} \le \left( \frac{\sqrt{2}-1}{2} \right)s.
\]
Isolating $s$ (and using that $s > 0$), we get that the above is equivalent to
\[\frac{2(3\sqrt{2}-4)}{(\sqrt{2}-1)^2} \le  s. \]
The LHS simplifies to $2\sqrt{2} \le 2.83$. Therefore, for any $s \ge 3 > 2\sqrt{2}$, the desired inequality holds, implying that the total expected welfare of all items from this component is upper bounded by
\[
s(\sqrt{2} +  \psi(\delta, p)).
\]
This completes the proof of Equation \eqref{eq:to_show} for exchange components when $s \ge 3$.

\paragraph{Small Exchanges ($s=2$).} The only non-degenerate case left is a 1-for-1 trade ($\ell=1, k=1$).

Consider the expected value of agent $1$ for the item they receive. Their expected value for their endowed item is $1+\delta/2$ (by~\Cref{eq:exp_value_single_draw}). Their expected gain from trade, if allowed by agent $2$, is $\E[\max(v_{1, 2} - v_{1, 1}, 0)] = \frac{\E[v_{1, 2} - v_{1, 1}] + \E[|v_{1, 2} - v_{1, 1}|]}{2} =  \frac{\E[|v_{1, 2} - v_{1, 1}|]}{2}$. Note that agent $2$ is willing to trade with probability $1/2$, independently of agent $1$'s value. Thus, the expected welfare gain for agent $1$ from trading is $ \frac{\E[|v_{1, 2} - v_{1, 1}|]}{4}$. We can compute the expected absolute gap between two draws from the distribution $\mathcal{D}$ as
\begin{align*}
    \E[|v_{1, 2} - v_{1, 1}|] &\le  2 p q \left(\frac{1 - q}{p} - 1 \right) + 2p(1 - p - q) \left(\frac{1 - q}{p} \right) + 2q(1 - p - q) \cdot 1  + \delta \\
    &= 2 - 2q^2 - 2p - 2pq + \delta \\
    &\leq 2(1 - q^2)+ \delta .
\end{align*}
The expected welfare for agent $1$'s item in this component is therefore upper bounded by  
\begin{align*}
    1 + \frac{\delta}{2} + \frac{1}{4}(2(1 - q^2) + \delta) &\le \frac{3 - q^2}{2} + \delta \\
    &= \frac{3 - (2 - \sqrt{2})^2}{2}+ \delta \\
    &= \frac{4\sqrt{2} - 3}{2} + \delta\\
    &= \sqrt{2} + \frac{2\sqrt{2} - 3}{2}+ \delta \\
    &\le \sqrt{2} + \delta.
\end{align*}

This shows that agent 1's expected welfare from the exchange is at most $\sqrt{2} + \delta$. By symmetry, the same holds for agent 2, yielding a $\sqrt{2} + \delta$ per-item bound on the final expected welfare for this component. This shows Equation \eqref{eq:to_show} holds for exchange components when $s = 2$.

Therefore, we have shown that Equation \eqref{eq:to_show} holds for all $s$ and all exchange mechanisms and for all picking mechanisms, and therefore it must hold for all components of any DSIC mechanism. This completes the proof of Equation \eqref{eq:overall_to_show}, and therefore the proof of~\Cref{thm:neg-2}.
\end{proof}

\paragraph{A Matching Bound.}
Strikingly, the upper bound of~\Cref{thm:neg-2} is tight. We can match the $\approx 0.854$ ratio using a remarkably simple class of rules we call \emph{Pick-$r$} mechanisms. For $r \in [0,1]$, a Pick-$r$ mechanism for two agents grants agent 1 their $\lfloor rm \rfloor$ most valuable items, leaving the rest for agent 2. This is simply a picking mechanism where the ``menu" consists of all subsets of a fixed cardinality; therefore, it is inherently DSIC.
\begin{proposition}\label{thm:pos-welf-2}
For two agents and any valid distribution $\mathcal{D}$, the Pick-$r$ mechanism, for $r = (\frac{2 - \sqrt{2}}{2})$, is DSIC and achieves a $\frac{2 + \sqrt{2}}{4} - o(1)$ approximation to welfare.
\end{proposition}

This establishes that simple picking mechanisms are, in a sense, the ``optimal'' DSIC mechanisms for two agents under i.i.d. valuations.

The proof of \Cref{thm:pos-welf-2} relies on the key observation that for a large number of items, a Pick-$r$ mechanism for two agents is approximately equivalent to the first agent taking any items for which its corresponding quantile (with respect to the underlying distribution) is at least $1 - r$. Therefore, rather than considering each distribution separately, we can instead analyze a Pick-$r$ mechanism via the quantile of the value of the agent who is assigned the item. Concretely, consider any one item. In the optimal allocation, the item is allocated to the agent with the highest value; therefore, the corresponding quantile is the maximum of two Uniform $[0, 1]$ random variables. To bound our approximation ratio, we compare this with the equivalent quantile for a Pick-$r$ mechanism.

The fact that Pick-$r$ mechanisms are approximately equivalent to this quantile version will be useful here, and especially when generalizing to $n > 2$ agents. We first formalize this connection.

\begin{definition}[\emph{Pick-$\mathbf{r}$} mechanisms]\label{dfn: pick r}
Given a vector $ \mathbf{r} = (r_1, r_2, \ldots, r_n)$ such that (i) $r_i \ge 0$ for all $i$, and (ii) $\sum_{i=1}^n r_i = 1$, let $t_i =  \lfloor m \cdot \sum_{j=1}^i r_j \rfloor - \lfloor m \cdot \sum_{j=1}^{i-1} r_j \rfloor$. The \emph{Pick-$\mathbf{r}$} mechanism is a \emph{quota serial dictatorship}~\cite{papai2000strategyproof} where agent $1$ selects their favorite $t_1$ items, then agent $2$ selects their favorite $t_2$ items from the remainder, and so on. (Note that this choice of $t_i$ ensures that $t_i \in \{\floor{r_i m}, \ceil{r_im} \}$ and $\sum_i t_i = m$.)    
\end{definition}

\begin{definition}[QT-$\mathbf{s}$ mechanism]\label{dfn: qts}
For a vector $\mathbf{s} = (s_1, \ldots, s_n)$ with $s_n = 1$, we define the quantile-threshold-$\mathbf{s}$ (QT-$\mathbf{s}$) mechanism as follows: For each item value $v_{i,j}$, we determine its quantile $q_{i,j} \in [0, 1]$.\footnote{Recall that we assumed that distributions are continuous, so this is well-defined.} The mechanism iterates through agents $1, \ldots, n$; item $j$ is allocated to the first agent $i$ such that $q_{i,j} \ge 1 - s_i$.
\end{definition}

We now show that these two mechanisms produce nearly identical allocations, and consequently, have similar guarantees. 
\begin{lemma}\label{lem:equiv}
    Fix a vector $\mathbf{r} = (r_1, \ldots, r_n)$ such that (i) $r_i \ge 0$ for all $i$, and (ii) $\sum_i r_i = 1$, and let $\mathbf{s} = (s_1, \ldots, s_n)$ be such that
    \[
        s_i = \frac{r_i}{1 - \sum_{k = 1}^{i - 1} r_k}.
    \]
    We have that
    \begin{enumerate}
        \item With probability $1 - O(n/m^2)$, the allocations of $\text{Pick-}\mathbf{r}$ and $\text{QT-}\mathbf{s}$ differ on at most $O(n^2 \sqrt{m \log m})$ items.
    \end{enumerate}
    Furthermore, fix a valid distribution $\mathcal{D}$ with CDF $F$, and let $X_1, \ldots,X_n$ and $X$ all be i.i.d.\@ draws from $\D$. 
    Let $\tau_i = 1 - s_i$.
    \begin{enumerate}[resume]
    \item $\text{Pick-}\mathbf{r}$ achieves an $\alpha - o(1)$ approximation to the optimal expected welfare, where
$$\alpha = \frac{\sum_{i = 1}^n r_i \cdot \E[X \mid X \ge F^{-1}(\tau_i)]}{\E[\max_i X_i]}.$$
\item$\text{Pick-}\mathbf{r}$ is Envy-Free (EF) with high probability if for all agents $i$ and $j$, if $i < j$, 
$$r_i \E[X \mid X \ge F^{-1}(\tau_i)] > r_j \E[X \mid  X \le F^{-1}(\tau_i)],$$ 
and if $i > j$, 
$$r_i\E[X \mid X \ge F^{-1}(\tau_i)] > r_j \E[X].$$
\end{enumerate}
\end{lemma}

The proof of \Cref{lem:equiv} can be found in Appendix \ref{sec:equiv}. Using \Cref{lem:equiv}, we are now ready to prove \Cref{thm:pos-welf-2}.

\begin{proof}[Proof of \Cref{thm:pos-welf-2} ]
    We directly analyze the corresponding QT mechanism from \Cref{lem:equiv}. Specifically, Pick-$r$, for $r=\tfrac{2-\sqrt{2}}{2}$, behaves near-identically to QT-$(\tfrac{2-\sqrt{2}}{2},1)$ (for a large number of items). Let $\tau_1 = 1-\tfrac{2-\sqrt{2}}{2} = \frac{\sqrt{2}}{2}$ and $\tau_2 = 0$.
    From part (2) of~\Cref{lem:equiv} we have that $\text{Pick-}\mathbf{r}$ achieves an $\alpha - o(1)$ approximation to the optimal expected welfare, where
\begin{align*}
    \alpha &= \frac{\sum_{i = 1}^n r_i \cdot \E[X \mid X \ge F^{-1}(\tau_i)]}{\E[\max_i X_i]} \\
    &= \frac{ \tfrac{2-\sqrt{2}}{2} \cdot \E[X \mid X \ge F^{-1}(\tfrac{\sqrt{2}}{2})] + (1 -\tfrac{2-\sqrt{2}}{2}) \cdot \E[X \mid X \ge F^{-1}(0)]}{\E[\max \{ X_1, X_2 \}]} \\
    &= \frac{ \tfrac{2-\sqrt{2}}{2} \cdot \E[X \mid X \ge F^{-1}(\tfrac{\sqrt{2}}{2})] + \tfrac{\sqrt{2}}{2} \cdot \E[X]}{\E[\max \{ X_1, X_2 \}]}\\
    &= \frac{ r \cdot \E[X \mid X \ge F^{-1}(1-r)] + (1-r)\cdot \E[X]}{\E[\max \{ X_1, X_2 \}]}
\end{align*}
    For the numerator we have that $r \cdot \E[X \mid X \ge F^{-1}(1-r)] = \E[X \cdot 1\{ X \geq F^{-1}(1-r) \}] =  \int_{1-r}^1 F^{-1}(x) dx$. We also have that $\E[X] = \int_{0}^{1} F^{-1}(x) dx$. Therefore, the numerator becomes 
\begin{align*}
  \int_{1-r}^1 F^{-1}(x) dx + (1-r) \int_{0}^{1} F^{-1}(x) dx &= \int_{1-r}^1 F^{-1}(x) dx + (1-r) \left( \int_{0}^{1-r} F^{-1}(x) dx + \int_{1-r}^{1} F^{-1}(x) dx \right) \\
  &= (1-r) \int_{0}^{1-r} F^{-1}(x) dx + (2-r) \int_{1-r}^1 F^{-1}(x) dx \\
  & = \int_{0}^1 g(x) F^{-1}(x) dx,
\end{align*}
where $g(x) = 1-r$ for $x < 1-r$, and $g(x) = 2-r$ for $x \geq 1-r$. Similarly, the denominator is equal to $\int_{0}^1 2x F^{-1}(x) dx$. Therefore, we have that
\[
\alpha = \frac{\int_{0}^1 g(x) F^{-1}(x) dx}{\int_{0}^1 2x F^{-1}(x) dx}.
\]
We will prove that $\int_{t}^1 g(x) dx \geq \frac{2+\sqrt{2}}{4} \cdot \int_{t}^1 2x dx$, for all $t \in [0,1]$, which implies that $\alpha \geq \frac{2+\sqrt{2}}{4}$ (since $F^{-1}(x)$ is a non-decreasing function), concluding the proof.
We have that $\int_{t}^1 2x dx = 1 - t^2$. For $t < 1-r$, $\int_{t}^1 g(x) dx = \int_{t}^{1-r} 1-r dx + \int_{1-r}^1 2-r dx = (1-r)(1-r-t) + (2-r)r = 1 - (1-r)t = 1 - \frac{\sqrt{2}}{2}t = \frac{2 - \sqrt{2} t}{2}$. It is easy to confirm that $\frac{2 - \sqrt{2} t}{2} \geq \frac{2+\sqrt{2}}{4} (1-t^2)$. For $t \geq 1-r$, $\int_{t}^1 g(x) dx = \int_{t}^{1} 2-r dx = (2-r)(1-t) = \frac{2+\sqrt{2}}{2}(1-t)$. It is easy to confirm that $\frac{2+\sqrt{2}}{2}(1-t) \geq \frac{2+\sqrt{2}}{4} (1-t^2)$. 

This concludes the proof.
\end{proof}

\paragraph{Incorporating Fairness.}

While \Cref{thm:pos-welf-2} provides a welfare approximation, it says nothing about fairness. Indeed, the parameter choice $r = \frac{2 - \sqrt{2}}{2} \approx 0.29$ implies that agent $2$ receives over 70\% of the items. Under distributions where items have similar values, agent $1$ will inevitably be envious. Furthermore, it is straightforward to show that varying $r$ in Pick-$r$ mechanisms to balance the allocation (and restore fairness) leads to suboptimal welfare guarantees.

One might therefore suspect a necessary trade-off between fairness and efficiency among DSIC mechanism. However, strikingly, we can incorporate envy-freeness at essentially no loss to the worst-case welfare guarantee. The key insight is that we don't need to restrict ourselves to a single picking mechanism for all distributions; instead, we show that for any distribution, there exists a choice of $r$ that simultaneously achieves both envy-freeness and optimal welfare guarantees.

\begin{theorem}\label{thm: main positive result for two agents}
    For two agents and any valid distribution $\mathcal{D}$, there exists an $r$ for which the Pick-$r$ mechanism is simultaneously DSIC, envy-free with high probability, and achieves a $\frac{2 + \sqrt{2}}{4} - o(1)$ approximation to welfare.
\end{theorem}

The proof of \Cref{thm: main positive result for two agents} builds on the result of \Cref{thm:pos-welf-2}, which showed that the Pick-$r$ mechanism for $r = \frac{2-\sqrt{2}}{2}$ achieves the desired welfare approximation. The key idea for the proof of \Cref{thm: main positive result for two agents} is separate all possible underlying distribution in to two categories. The first category is the set of distributions for which the Pick-$\tfrac{2-\sqrt{2}}{2}$ is envy-free with high probability. For these distributions,  \Cref{thm:pos-welf-2} directly implies the desired result. The second category is the set of distributions for which the  Pick-$\tfrac{2-\sqrt{2}}{2}$ mechanism is not envy-free with high probability. Here, we show that for sufficiently small $\varepsilon > 0$, the Pick-$(1/2-\varepsilon)$ mechanism achieves at least a $7/8-o(1) \ge \tfrac{2+\sqrt{2}}{4} - o(1)$ approximation to welfare and is envy-free with high probability. Therefore, in this case, the Pick-$(1/2 - \varepsilon)$ satisfies the desired property. The proof of \Cref{thm: main positive result for two agents} can be found in Appendix \ref{sec: main positive result for two agents}.

\subsection{Connection to Prophet Inequalities}

For a reader that is well-versed in optimal stopping theory, the ratio $\frac{2 + \sqrt{2}}{4} \approx 0.854$ from \Cref{thm:pos-welf-2} may look familiar. This is no coincidence. It turns out that there is a tight connection between i.i.d.\@ prophet inequalities and the approximation ratios of a class Pick-$\mathbf{r}$ (or, more specifically, QT-$\mathbf{s}$ mechanisms).
To discuss this, we will need to formally introduce the concepts and give a brief background of the area.

\paragraph{Preliminaries on I.I.D. Prophets Inequalities.} In the i.i.d.\@ prophet setting, a sequence of $n$ random variables $X_1, \ldots, X_n$ are drawn i.i.d.\@ from a fixed, known distribution $\mathcal{D}$. An online algorithm observes the realized values $x_1, \ldots, x_n$ sequentially. Upon observing $x_i$, the algorithm must make an irrevocable decision to either accept $x_i$ and stop, or reject $x_i$ and continue to step $i+1$. An algorithm in this problem is characterized by a \emph{stopping time} $\tau$, a random variable taking values in $\{1, \ldots, n\}$ such that the event $\{\tau = i\}$ depends only on the history $X_1, \ldots, X_i$. The objective is to maximize the expected reward, $\E[X_\tau]$. The algorithm competes against a ``prophet'' who knows the realizations of all random variables in advance and selects the maximum. We say an algorithm achieves an $\alpha$-approximation (or has a competitive ratio of $\alpha$) if $\E[X_\tau] \ge \alpha \cdot \E[\max_i X_i]$.

This problem was first studied by \citet{hill1982comparisons} (while the non-i.i.d. version dates back to~\citet{krengel1978semiamarts}), who derived a recursive formula for the optimal $\alpha$ for any fixed $n$, denoted by $a_n$. For $n = 2$, this bound is $\frac{2 + \sqrt{2}}{4} \approx 0.854$.\footnote{Hill and Kertz originally analyzed the ratio of the prophet to the gambler ($1/\alpha$). We adopt the modern convention of the approximation ratio $\alpha \le 1$.} They further demonstrated that for all $n$, $a_n \ge 1 - 1/e \approx 0.632$. Subsequently, \cite{kertz1986stop} proved that as $n \to \infty$, $a_n$ converges to a constant $\beta \approx 0.745$. This constant is defined such that $1/\beta$ satisfies the integral equation:$$\int_0^1 \frac{1}{y(1 - \ln y) + \beta - 1} \, dy = 1.$$ While Kertz established $\beta$ as the asymptotic limit, he did not prove that the sequence of $a_n$s is monotone, leaving open the possibility that the bound for finite $n$ could be strictly worse than $\beta$. This gap was closed by \cite{correa2017posted}, who proved that $\beta \approx 0.745$ is a universal lower bound for all $n \ge 1$.

We can now state the equivalence between i.i.d. prophet inequalities and DSIC mechanisms in our setting. In particular, we will define the Quantile-threshold-prophet-$\mathbf{s}$ (QTP-$\mathbf{s}$) algorithm for the i.i.d. prophet inequality problem, which, upon observing value $x_i$, it accepts it if and only if $q_i \ge 1 - s_i$, where $q_i$ is the quantile of $x_i$ (i.e., $q_i = F(x_i)$).
We then have the following.

\begin{proposition}\label{thm:prophet-connection}
    Fix a valid distribution $\mathcal{D}$ and a number of agents $n$. Fix $\alpha \in [0, 1]$. The following statements hold.
    \begin{enumerate}
        \item For any $\mathbf{s}$, the QT-$\mathbf{s}$ algorithm (\Cref{dfn: qts}) achieves an $\alpha$-approximation to welfare iff the QTP-$\mathbf{s}$ algorithm achieves an $\alpha$-approximation to the reward of the prophet.
        \item There exists a prophet algorithm achieving an $\alpha$-approximation iff there exists a vector $\mathbf{r}$ for which Pick-$\mathbf{r}$ achieves an $\alpha$-approximation.
    \end{enumerate} 
\end{proposition}
\begin{proof}
    For $(1)$, fix an arbitrary $\mathbf{s}$. Consider a single item $j$ with quantiles $q_{i,j} = F(v_{i,j})$. By Definition~\ref{dfn: qts}, QT-$\mathbf{s}$ allocates $j$ to $\tau = \min\{ i \in [n] : q_{i,j} \ge 1-s_i \}$. The QTP-$\mathbf{s}$ stopping rule is exactly the same, hence on every realization QT-$\mathbf{s}$ and QTP-$\mathbf{s}$ select the same index $\tau$ and obtain the same realized value. Summing over items, the total welfare of QT-$\mathbf{s}$ equals the total reward of running QTP-$\mathbf{s}$ independently on each item (with the same $\mathbf{s}$), and therefore their expectations coincide. Finally, the offline benchmark in our allocation setting is $\sum_j \max_i v_{i,j}$ (each item assigned to the agent with the highest value), which is exactly the prophet benchmark for these per-item instances.

    For $(2)$, first note that by~\Cref{lem:equiv}, every Pick-$\mathbf{r}$ mechanism is equivalent to a QT-$\mathbf{s}$ mechanism (with respect to the welfare approximation, up to a $o(1)$ difference), for some threshold vector $\mathbf{s}$. Second, by part (1) of this proposition, the welfare of a QT-$\mathbf{s}$ mechanism is identical to the reward of a prophet algorithm QTP-$\mathbf{s}$. Therefore, it suffices to show that the optimal prophet algorithm must be a QTP-$\mathbf{s}$ algorithm.

    Consider an arbitrary optimal prophet algorithm, and let $z_i$ be the marginal probability that the algorithm stops at step $i$. To maximize the expected reward $\sum_i \E[X_i | \text{stop at i} ] \cdot z_i$, for these (fixed) marginal stopping probabilities, the optimal strategy at step $i$ (conditioned on reaching step $i$) must be to accept when $X_i$ takes a large value. This implies a threshold strategy where we accept $X_i$ if and only if it exceeds a certain quantile $s_i$, and specifically $s_i = z_i / Pr(\text{reach } i)$. This is a QTP-$\mathbf{s}$ strategy, and it achieves the same welfare as the original algorithm. Therefore, applying part (1) of this proposition and ~\Cref{lem:equiv}, we have a correspondence between optimal prophet algorithms and Pick-$\mathbf{r}$ mechanisms.
\end{proof}

\subsection{$n > 2$ Agents}\label{sec:n-agents}
We now turn to the case of $n > 2$ agents.
By importing results from the i.i.d.\@ prophet problem~\cite{correa2017posted}, we can ensure that, for any distribution $\mathcal{D}$ and any number of agents $n$, there exists a mechanism achieving a $\beta \approx 0.745$ approximation to welfare. 
\begin{corollary}\label{thm:correa-pos-n}
    For all $n$ and distributions $\mathcal{D}$, there exists $\mathbf{r} = (r_1, \ldots, r_{n-1})$ such the Pick-$\mathbf{r}$ mechanism achieves a $\beta - o(1) \approx 0.745$ approximation to welfare. Furthermore, there exists a valid distribution $\mathcal{D}$ such that no picking sequence can achieve a higher approximation to welfare. 
\end{corollary}
\begin{proof}
    By \Cref{thm:prophet-connection}, the welfare approximation of a Pick-$\mathbf{r}$ mechanism is equivalent to the competitive ratio of an i.i.d. prophet inequality algorithm. Since $\beta - o(1) \approx 0.745$ is the optimal bound for the i.i.d. prophet inequality problem~\cite{correa2017posted}, the existence of the corresponding Pick-$\mathbf{r}$ mechanism follows.
\end{proof}

\emph{Incorporating Fairness.}
The most natural approach to fairness in this framework is to equalize the number of items agents pick, i.e., to use the Pick-$(1/n, \ldots, 1/n)$ mechanism. This corresponds to a stopping rule studied in the prophet inequality literature ~\cite{arsenis2022individual}, albeit under a different notion of fairness. For this mechanism, we can prove the following guarantee: for all $n$, the Pick-$(1/n, \ldots, 1/n)$ mechanism is DSIC, envy-free in expectation, and achieves a $1/2$ approximation to welfare. By slightly modifying the number of items allocated to each agent, we are able to find a mechanism that can instead achieve envy-freeness with high probability.

\begin{theorem}\label{thm: n agents dsic 2 approx}
    For all $n$ and all valid distributions $\mathcal{D}$, for sufficiently small $\varepsilon$, the Pick-$\mathbf{r}$ mechanism with $\mathbf{r} = (1/n - \varepsilon, \cdots, 1/n - \varepsilon, 1/n + (n -1)\varepsilon)$ is envy-free with high probability and achieves a $1/2 - o(1)$ approximation to welfare. 
\end{theorem}
\begin{proof}[Proof of \Cref{thm: n agents dsic 2 approx}]
    To show the welfare approximation, we will apply part $(2)$ of \Cref{lem:equiv} to show that the welfare approximation is strictly greater than $1/2$ for $\varepsilon = 0$. Since the welfare approximation is continuous in $\varepsilon$ (as the thresholds in the expected values are continuous), it will remain greater than $1/2$ for sufficiently small $\varepsilon > 0$. 
    
    Using the definition in~\Cref{lem:equiv}, we have that $s_i = \frac{1/n}{1 - (i-1)/n} = \frac{1}{n - i + 1}$. Therefore, using part $(2)$ of \Cref{lem:equiv}, we have that the approximation ratio for $\varepsilon = 0$ is
    \[
    \frac{\sum_{i = 1}^n \frac{1}{n} \cdot \E[X \mid X \ge F^{-1}\left( 1 - \frac{1}{n + 1 - i} \right)]}{\E[\max_i X_i]} = \frac{\sum_{i = 1}^n \frac{1}{n} \cdot \E[X \mid X \ge F^{-1}\left(1 - \frac{1}{i} \right)]}{\E[\max_i X_i]}.
    \]

    Using the identity $\E[Y] = \int_0^\infty \Pr[Y > x]\, dx$ for nonnegative random variables, we can rewrite the expectation terms. For a fixed integration variable $x$, the conditional probability $\Pr[X > x \mid X \ge F^{-1}(1 - 1/i)]$ depends on whether $x$ is below or above the threshold $F^{-1}(1 - 1/i)$.
    \begin{itemize}
        \item If $x \le F^{-1}(1 - 1/i)$, then the condition $X > x$ is always met given the condition $X \ge F^{-1}(1 - 1/i)$, and thus, the probability is $1$.
        \item If $x > F^{-1}(1 - 1/i)$, then $\Pr[X > x \mid X \ge F^{-1}(1 - 1/i)] = \frac{\Pr[X > x]}{\Pr[X \ge F^{-1}(1 - 1/i)]} = \frac{1 - F(x)}{1/i} = i(1 - F(x))$.
    \end{itemize}
    
    Therefore, $\Pr[X > x \mid X \ge F^{-1}(1 - 1/i)] = \min\{ 1 , i(1 - F(x)) \}$. Summing over $i$ we have that $\sum_{i = 1}^n \frac{1}{n} \cdot \E[X \mid X \ge F^{-1}\left(1 - \frac{1}{i} \right)]$ is equal to
    \[
     \sum_{i=1}^n \frac{1}{n} \int_{x=0}^\infty \min\{ 1 , i(1 - F(x)) \} dx = \int_{x=0}^\infty \left( \frac{1}{n} \sum_{i=1}^n \min\{ 1 , i(1 - F(x)) \} \right) dx.
    \]
    
    On the other hand, $\E[\max_i X_i] = \int_0^\infty (1 - F(x)^n)\, dx$. We will prove that $\frac{1}{n} \sum_{i=1}^n \min\{ 1 , i(1 - F(x)) \} > \frac{1}{2} (1 - F(x)^n) $ for all $x \geq 0$.

    \paragraph{Case 1: $x > F^{-1}(1 - 1/n)$.} In this case, $1 - F(x) < 1/n$, and therefore $i(1 - F(x)) < 1$ for all $i \in [n]$. Therefore, 
    \[
    \frac{1}{n} \sum_{i=1}^n \min\{ 1 , i(1 - F(x)) \} = \frac{1 - F(x)}{n} \sum_{i=1}^n i = \frac{1-F(x)}{n} \frac{n(n+1)}{2} = \frac{n+1}{2} (1-F(x)).
    \]
    Using Bernoulli's inequality, $F(x)^n = (1 - (1 - F(x)))^n \ge 1 - n(1 - F(x))$. This implies that $1 - F(x)^n \le n(1 - F(x))$. So we overall have
    \[
    \frac{1}{n} \sum_{i=1}^n \min\{ 1 , i(1 - F(x)) \} = \frac{n+1}{2} (1-F(x)) > \frac{n}{2}(1 - F(x)),
    \]
    as desired.
    \paragraph{Case 2: $x \leq F^{-1}(1 - 1/n)$.}
    Here we use the trivial bound $1 - F(x)^n \le 1$, so it suffices to show that $\frac{1}{n} \sum_{i=1}^n \min\{ 1 , i(1 - F(x)) \} > 1/2$. Since $x \leq F^{-1}(1 - 1/n)$, $1 - F(x) \geq 1/n$. Therefore,
    \[
    \frac{1}{n} \sum_{i=1}^n \min\{ 1 , i(1 - F(x)) \} \geq \frac{1}{n} \sum_{i=1}^n \min\{ 1 , i/n \} = \frac{1}{n} \frac{n(n+1)}{2n} > \frac{1}{2},
    \]
    as desired.

    \paragraph{Envy-Freeness:}
We now verify envy-freeness for the vector $\mathbf{r} = (1/n - \varepsilon, \dots, 1/n - \varepsilon, 1/n + (n - 1)\varepsilon)$.
Let $\tau_i$ be the value threshold for agent $i$. For $i < n$, $\tau_i \in (0, 1)$ because $s_i < 1$. For agent $n$, $s_n=1$ (since they take everything remaining), so $\tau_n = F^{-1}(0) = 0$. There are three cases based on whether or not one of the agents is agent $n$.
\begin{enumerate}
    \item \emph{Agents $i, j < n$.} Here, $r_i = r_j$. Since $X$ is continuous and $\tau_i > 0$, $\E[X \mid X \ge F^{-1}(\tau_i)] > \E[X] \ge \E[X \mid X < F^{-1}(\tau_i)]$, so the inequality holds regardless of the order of $i$ and $j$.
    \item \emph{Agent $j = n$.} Here $r_i = 1/n - \varepsilon$ and $r_j = 1/n + (n-1)\varepsilon$. However, $\E[X \mid X \ge F^{-1}(\tau_i)] > \E[X \mid X < F^{-1}(\tau_i)]$. Thus, for sufficiently small $\varepsilon$, $r_i \E[X \mid X \ge F^{-1}(\tau_i)] > r_j \E[X \mid X < F^{-1}(\tau_i)]$.
    \item \emph{Agent $i = n$.} Here $r_i = 1/n + (n - 1)\varepsilon$. Furthermore, $\tau_i = 0$, so $\E[X \mid X \ge F^{-1}(\tau_i)] = \E[X]$. Therefore, for all $\varepsilon > 0$, $r_i \E[X \mid X \ge F^{-1}(\tau_i)] > r_j \E[X]$. 
\end{enumerate}
    Since all expected value conditions strictly hold, by \Cref{lem:equiv} part (3), the mechanism is envy-free with high probability.
\end{proof}

\section{BIC Mechanisms}

In this section, we study Bayesian Incentive Compatible mechanisms. Consider the \emph{Ranking} algorithm, defined as~\Cref{algo:ranking}. The algorithm, at a high level, works as follows. We interpret the agents' bids as ordinal, i.e., we convert the cardinal (reported) values to a ranking over the items. Then, we give each item to the agent who ranked it highest. We prove that this achieves strong theoretical guarantees.

Specifically, the algorithm will compute the ranking $\sigma_i$ based on the bids $b_i$ for each agent $i$, where item $j$ is ranked higher than item $k$ in $\sigma_i$ if and only if $b_{i,j} \geq b_{i,k}$. The algorithm will then give each item to the agent with the highest rank for that item (with consistent tie-breaking). Let $\rank(\sigma_i, j)$ be the rank of item $j$ in agent $i$'s list, with $1$ being the top ranked item. 

\begin{algorithm}[htb!]
\caption{Ranking Algorithm}
\label{algo:ranking}
\KwIn{Bids $\{b_{i,j}\}_{i \in \mathcal{N}, j \in \mathcal{M}}$ for each agent, item pair}
\BlankLine
$\sigma_i \gets \text{Ranking of bids $b_i$, breaking ties arbitrarily,} \quad  \forall i$\\
$A_i \gets \emptyset, \quad \forall i$

\For{$j \gets 1$ \KwTo $m$}{
    $r \gets \min_{i} \rank(\sigma_i, j)$ \\
    $i^* \gets \min_{i:\,\rank(\sigma_i, j)=r} i$ \\
    $A_{i^*} \gets A_{i^*} \cup \{j\}$ \\
}
\end{algorithm}

\begin{theorem}\label{thm: bic iid theorem}
    For all valid distributions $\mathcal{D}$, the ranking algorithm (\Cref{algo:ranking}) is BIC, envy-free with high probability, and achieves a $1 - o(1)$ approximation to welfare. 
\end{theorem}

To prove the BIC property, we show that for each position $j$, there is a fixed probability $p_j$ that an agent receives the item ranked in position $j$. Furthermore, these probabilities are monotonically decreasing in $j$. Thus, to maximize their expected value for the items they receive, an agent should report their true ranking, which is induced by truthfully reporting their values.
Regarding fairness and efficiency, the i.i.d.\@ assumption implies that~\Cref{algo:ranking} approximates \emph{welfare maximization}, i.e., the rule that gives each item to the agent with the highest value. Welfare maximization is known to have strong fairness and efficiency guarantees in this setting~\cite{dickerson2014computational}. Our proof shows that these guarantees extend to this approximate regime.

Consequently,~\Cref{algo:ranking} achieves nearly all desired properties. The primary trade-off is the $1 - o(1)$ approximation rather than exact optimality. However, it is straightforward to show that any mechanism achieving a perfect welfare guarantee cannot be incentive compatible: optimal welfare requires allocating every item to the agent with the highest value, creating an unavoidable incentive for agents to exaggerate their bids.

\begin{proof}[Proof of \Cref{thm: bic iid theorem}]
First, we prove the BIC property. Since item values are i.i.d.\ across agents, the rankings $\sigma_{-i}$ induced by other agents are, from $i$'s perspective, exchangeable across items: for any fixed item, the relative ranks assigned to it by the other agents have the same distribution as for any other item. Consequently, for each position $t \in [m]$ there exists a number $p_t \in [0,1]$ such that, regardless of which item agent $i$ places in position $t$, the probability that $i$ wins that item equals $p_t$. Moreover, these probabilities are monotone: $p_1\ge p_2\ge \cdots \ge p_m$, since moving an item higher in $\sigma_i$ can only increase $i$'s probability of being the highest-ranked agent for that item (tie-breaking is fixed and consistent). Therefore, given $\sigma_i$ and true valuation $v_i$, agent $i$'s expected utility can be written as
\[
\E[ u_i | \sigma_i, v_i ] = \sum_{t=1}^m p_t v_{i,\pi(t)},
\]
where $\pi(t)$ is the item agent $i$ placed in position $t$ (according to bids $b_i$). Since $(p_t)_{t=1}^m$ is a fixed, non-increasing sequence, the above expression is maximized by assigning larger values to smaller positions. Thus, agent $i$ maximizes expected utility by submitting bids whose rankings order the items by decreasing true value (i.e., the ranking induced by truthful bidding). This establishes BIC.

Next, we show that~\Cref{algo:ranking} matches welfare maximization on up to $O(n\sqrt{m \log m})$ items with probability $1 - O(n/m^2)$.

To this end, let $\hat{F}_i$ be $i$'s empirical distribution of values. Let $\mathcal{E}_i$ be the event that $\sup_x |\hat{F}_i(x) - F(x)| \le \sqrt{\log m / m}$. By the DKW inequality~\cite{massart1990tight}, 
\[
    \Pr[\mathcal{E}_i] \ge 1 - 2\exp(-2m(\sqrt{\log m / m})^2) = 1 - 2\exp(-2\log m) = 1 - 2/m^2.
\] Conditioned on $\mathcal{E}_i$ for all $i \in \mathcal{N}$, all empirical quantiles are within $\sqrt{\log m / m}$ of their true quantiles. Thus, as long as there is a quantile difference between agents of at least $2\sqrt{\log m / m}$, then the higher-value agent will rank it higher. The difference between the highest quantile and second highest quantile of $n$ i.i.d.\@ draws follows a $\text{Beta}(1, n)$ distribution. As the PDF of this distribution is at most $n$, the probability that it takes a value smaller than $2\sqrt{\log m / m}$ is at most $2n\sqrt{\log m / m}$ for each item, and this is independent across items. The expected number of items that have a small highest-to-second-highest-quantile difference (recall that only those items can be allocated to the ``wrong'' agent) is therefore $2n\sqrt{m \log m}$. By Hoeffding's inequality, the number of items that have a small highest-to-second-highest-quantile difference is at most $(2n + 1)\sqrt{m \log m}$ with probability at least $1 - \exp(-2(\sqrt{m \log m})^2/m) = 1- 1/m^2$. Conditioned on this event along with all of $\mathcal{E}_i$ (which by a union bound occurs with probability $1 - O(n/m^2)$), \Cref{algo:ranking} and welfare maximization differ on at most $O(n \sqrt{m \log m})$ items. This implies a $1 - o(1)$ welfare approximation. Furthermore, by~\citet{dickerson2014computational}, for the allocation $A^*$ chosen by the welfare maximization algorithm, $v_i(A^*_i) - v_i(A^*_j) \in \Omega(m/n)$ with high probability (recalling that values are in $[0,1]$ without loss of generality). Thus, conditioned on this event as well, even moving $O(n\sqrt{m \log m})$ items adversarially, the allocation remains envy-free with high probability.
\end{proof}

Finally, we show that~\Cref{algo:ranking} can be extended to a large class of independent but \emph{non}-identical distributions (e.g., different agents have different distributions) using similar tools to~\cite{baienvy2022} (who study the existence of envy-free and efficient allocations in the non-i.i.d. stochastic setting). We expand on this in \Cref{app:non-identical-independent}.

\section{Discussion}
In this work, we revisited the problem of truthful fair division through the lens of stochastic valuations. We showed that moving beyond worst-case analysis allows us to bypass various impossibility results. 
Nonetheless, there are a number of avenues for future work. 

\paragraph{Tightness for $n > 2$ Agents.}
While we provide a complete picture for the DSIC two-agent case, the case for $n > 2$ has some gaps. Our prophet inequality-based mechanisms achieve a $\approx 0.745$-approximation, but we do not know if this is tight. Proving impossibility results for $n > 2$ is significantly more challenging than for the two-agent case. For $n=2$, we leveraged the structural characterization of \citet{amanatidis2017truthful} which restricts all DSIC mechanisms to "picking-exchange" rules. No such characterization exists for $n > 2$. In the general case, DSIC mechanisms can be much more complex; for instance, the mechanism used to divide items between agents 1 and 2 could depend on the report of agent 3 in intricate ways. Without a structural foothold, ruling out the existence of high-welfare DSIC mechanisms for more than two agents remains a difficult theoretical hurdle.

\paragraph{Correlated Valuations.} Our analysis relied heavily on the independence of item values. A natural and important direction for future work is to extend these results to distributions where values are correlated across agents (although still independent across items), for example, from \citet{benade2024fair}. One promising angle is to view the distribution of item values as a cake in the classic cake-cutting problem. For two agents, this analogy can be used to show that a Pick-$r$ style mechanism for some ordering of the agents achieves envy-freeness with high probability. We expand on this in \Cref{app:cake-cutting}. However, considering Pareto optimality (a more natural notion of efficiency when agents' value magnitudes are different) or extending this to more agents seems to require new techniques.

\paragraph{Beyond additive valuations.} Finally, it would be very compelling to extend these results to settings beyond additive valuations. In many combinatorial allocation problems, agents exhibit substitutes or complements (e.g., submodular or supermodular valuations). A major challenge here is modeling: defining natural, analytically tractable distributions over complex valuation functions is non-trivial. One recent example in this direction is the work of~\citet{benade2024existence}, which studies the existence of envy-free allocations under a beyond-additive stochastic model. However, incorporating incentives in such non-additive stochastic settings remains an open problem.

\bibliographystyle{ACM-Reference-Format}
\bibliography{references}

@article{hill1982comparisons,
  title={Comparisons of stop rule and supremum expectations of iid random variables},
  author={Hill, Theodore P and Kertz, Robert P},
  journal={The Annals of Probability},
  pages={336--345},
  year={1982},
  publisher={JSTOR}
}

@inproceedings{correa2017posted,
  title={Posted price mechanisms for a random stream of customers},
  author={Correa, Jos{\'e} and Foncea, Patricio and Hoeksma, Ruben and Oosterwijk, Tim and Vredeveld, Tjark},
  booktitle={Proceedings of the 2017 ACM Conference on Economics and Computation},
  pages={169--186},
  year={2017}
}

@article{arsenis2022individual,
  title={Individual fairness in prophet inequalities},
  author={Arsenis, Makis and Kleinberg, Robert},
  journal={arXiv preprint arXiv:2205.10302},
  year={2022}
}

@inproceedings{halpern2020fair,
  title={Fair division with binary valuations: One rule to rule them all},
  author={Halpern, Daniel and Procaccia, Ariel D and Psomas, Alexandros and Shah, Nisarg},
  booktitle={International Conference on Web and Internet Economics},
  pages={370--383},
  year={2020},
  organization={Springer}
}

@article{kertz1986stop,
  title={Stop rule and supremum expectations of iid random variables: a complete comparison by conjugate duality},
  author={Kertz, Robert P},
  journal={Journal of multivariate analysis},
  volume={19},
  number={1},
  pages={88--112},
  year={1986},
  publisher={Elsevier}
}

@inproceedings{halpern2025online,
  title={Online Envy Minimization and Multicolor Discrepancy: Equivalences and Separations},
  author={Halpern, Daniel and Psomas, Alexandros and Verma, Paritosh and Xie, Daniel},
  booktitle={Proceedings of the 26th ACM Conference on Economics and Computation},
  pages={188--188},
  year={2025}
}

@article{benade2024fair,
  title={Fair and efficient online allocations},
  author={Benad{\`e}, Gerdus and Kazachkov, Aleksandr M and Procaccia, Ariel D and Psomas, Alexandros and Zeng, David},
  journal={Operations Research},
  volume={72},
  number={4},
  pages={1438--1452},
  year={2024},
  publisher={INFORMS}
}

@article{krengel1978semiamarts,
  title={On semiamarts, amarts, and processes with finite value},
  author={Krengel, Ulrich and Sucheston, Louis},
  journal={Probability on Banach spaces},
  volume={4},
  number={197-266},
  pages={1--2},
  year={1978},
  publisher={Dekker New York}
}

@inproceedings{benade2024existence,
  title={On the Existence of Envy-Free Allocations Beyond Additive Valuations},
  author={Benad{\`e}, Gerdus and Halpern, Daniel and Psomas, Alexandros and Verma, Paritosh},
  booktitle={Proceedings of the 25th ACM Conference on Economics and Computation},
  pages={1287--1287},
  year={2024}
}

@article{manurangsi2021closing,
  title={Closing gaps in asymptotic fair division},
  author={Manurangsi, Pasin and Suksompong, Warut},
  journal={SIAM Journal on Discrete Mathematics},
  volume={35},
  number={2},
  pages={668--706},
  year={2021},
  publisher={SIAM}
}

@inproceedings{mossel2010truthful,
  title={Truthful fair division},
  author={Mossel, Elchanan and Tamuz, Omer},
  booktitle={International Symposium on Algorithmic Game Theory},
  pages={288--299},
  year={2010},
  organization={Springer}
}

@article{papai2001strategyproof,
  title={Strategyproof and nonbossy multiple assignments},
  author={P{\'a}pai, Szilvia},
  journal={Journal of Public Economic Theory},
  volume={3},
  number={3},
  pages={257--271},
  year={2001},
  publisher={Wiley Online Library}
}

@article{massart1990tight,
  title={The tight constant in the Dvoretzky-Kiefer-Wolfowitz inequality},
  author={Massart, Pascal},
  journal={The annals of Probability},
  pages={1269--1283},
  year={1990},
  publisher={JSTOR}
}

@inproceedings{babaioff2025truthful,
  title={On Truthful Mechanisms without Pareto-efficiency: Characterizations and Fairness},
  author={Babaioff, Moshe and Manaker Morag, Noam},
  booktitle={Proceedings of the 26th ACM Conference on Economics and Computation},
  pages={447--447},
  year={2025}
}

@article{ortega2022obvious,
  title={Obvious manipulations in cake-cutting},
  author={Ortega, Josu{\'e} and Segal-Halevi, Erel},
  journal={Social Choice and Welfare},
  volume={59},
  number={4},
  pages={969--988},
  year={2022},
  publisher={Springer}
}

@article{mennle2021partial,
  title={Partial strategyproofness: Relaxing strategyproofness for the random assignment problem},
  author={Mennle, Timo and Seuken, Sven},
  journal={Journal of Economic Theory},
  volume={191},
  pages={105144},
  year={2021},
  publisher={Elsevier}
}

@inproceedings{amanatidis2023round,
  title={Round-robin beyond additive agents: Existence and fairness of approximate equilibria},
  author={Amanatidis, Georgios and Birmpas, Georgios and Lazos, Philip and Leonardi, Stefano and Reiffenh{\"a}user, Rebecca},
  booktitle={Proceedings of the 24th ACM Conference on Economics and Computation},
  pages={67--87},
  year={2023}
}

@inproceedings{abebe2020truthful,
  title={A truthful cardinal mechanism for one-sided matching},
  author={Abebe, Rediet and Cole, Richard and Gkatzelis, Vasilis and Hartline, Jason D},
  booktitle={Proceedings of the fourteenth annual ACM-SIAM symposium on discrete algorithms},
  pages={2096--2113},
  year={2020},
  organization={SIAM}
}

@article{amanatidis2024allocating,
  title={Allocating indivisible goods to strategic agents: Pure nash equilibria and fairness},
  author={Amanatidis, Georgios and Birmpas, Georgios and Fusco, Federico and Lazos, Philip and Leonardi, Stefano and Reiffenh{\"a}user, Rebecca},
  journal={Mathematics of operations research},
  volume={49},
  number={4},
  pages={2425--2445},
  year={2024},
  publisher={INFORMS}
}

@inproceedings{babaioff2021fair,
  title={Fair and truthful mechanisms for dichotomous valuations},
  author={Babaioff, Moshe and Ezra, Tomer and Feige, Uriel},
  booktitle={Proceedings of the AAAI Conference on Artificial Intelligence},
  volume={35},
  number={6},
  pages={5119--5126},
  year={2021}
}

@article{amanatidis2021maximum,
  title={Maximum Nash welfare and other stories about EFX},
  author={Amanatidis, Georgios and Birmpas, Georgios and Filos-Ratsikas, Aris and Hollender, Alexandros and Voudouris, Alexandros A},
  journal={Theoretical Computer Science},
  volume={863},
  pages={69--85},
  year={2021},
  publisher={Elsevier}
}

@inproceedings{friedman2019fair,
  title={Fair and efficient memory sharing: Confronting free riders},
  author={Friedman, Eric J and Gkatzelis, Vasilis and Psomas, Christos-Alexandros and Shenker, Scott},
  booktitle={Proceedings of the AAAI Conference on Artificial Intelligence},
  volume={33},
  number={01},
  pages={1965--1972},
  year={2019}
}

@inproceedings{markakis2011worst,
  title={On worst-case allocations in the presence of indivisible goods},
  author={Markakis, Evangelos and Psomas, Christos-Alexandros},
  booktitle={International Workshop on Internet and Network Economics},
  pages={278--289},
  year={2011},
  organization={Springer}
}

@inproceedings{amanatidis2016truthful,
  title={On truthful mechanisms for maximin share allocations},
  author={Amanatidis, Georgios and Birmpas, Georgios and Markakis, Evangelos},
  booktitle={Proceedings of the Twenty-Fifth International Joint Conference on Artificial Intelligence},
  pages={31--37},
  year={2016}
}

@inproceedings{cole2013mechanism,
  title={Mechanism design for fair division: allocating divisible items without payments},
  author={Cole, Richard and Gkatzelis, Vasilis and Goel, Gagan},
  booktitle={Proceedings of the fourteenth ACM conference on Electronic commerce},
  pages={251--268},
  year={2013}
}

@inproceedings{barman2022truthful,
  title={Truthful and fair mechanisms for matroid-rank valuations},
  author={Barman, Siddharth and Verma, Paritosh},
  booktitle={Proceedings of the AAAI Conference on Artificial Intelligence},
  volume={36},
  number={5},
  pages={4801--4808},
  year={2022}
}

@inproceedings{ghodsi2011dominant,
  title={Dominant resource fairness: Fair allocation of multiple resource types},
  author={Ghodsi, Ali and Zaharia, Matei and Hindman, Benjamin and Konwinski, Andy and Shenker, Scott and Stoica, Ion},
  booktitle={8th USENIX symposium on networked systems design and implementation (NSDI 11)},
  year={2011}
}

@inproceedings{baienvy2022,
  title={Envy-Free and Pareto-Optimal Allocations for Agents with Asymmetric Random Valuations},
  author={Bai, Yushi and G{\"o}lz, Paul},
  booktitle={Proceedings of the Thirty-First International Joint Conference on Artificial Intelligence, IJCAI},
  year={2022}
}

@article{parkes2015beyond,
  title={Beyond dominant resource fairness: Extensions, limitations, and indivisibilities},
  author={Parkes, David C and Procaccia, Ariel D and Shah, Nisarg},
  journal={ACM Transactions on Economics and Computation (TEAC)},
  volume={3},
  number={1},
  pages={1--22},
  year={2015},
  publisher={ACM New York, NY, USA}
}

@inproceedings{friedman2014strategyproof,
  title={Strategyproof allocation of discrete jobs on multiple machines},
  author={Friedman, Eric and Ghodsi, Ali and Psomas, Christos-Alexandros},
  booktitle={Proceedings of the fifteenth ACM conference on Economics and computation},
  pages={529--546},
  year={2014}
}

@inproceedings{dickerson2014computational,
  title={The computational rise and fall of fairness},
  author={Dickerson, John and Goldman, Jonathan and Karp, Jeremy and Procaccia, Ariel and Sandholm, Tuomas},
  booktitle={Proceedings of the AAAI conference on artificial intelligence},
  volume={28},
  number={1},
  year={2014}
}

@inproceedings{hartman2025s,
  title={It’s Not All Black and White: Degree of Truthfulness for Risk-Avoiding Agents},
  author={Hartman, Eden and Segal-Halevi, Erel and Tao, Biaoshuai},
  booktitle={Proceedings of the 26th ACM Conference on Economics and Computation},
  pages={996--1016},
  year={2025}
}

@inproceedings{bu2024truthful,
  title={Truthful and Almost Envy-Free Mechanism of Allocating Indivisible Goods: the Power of Randomness},
  author={Bu, Xiaolin and Tao, Biaoshuai},
    booktitle={Proceedings of the 66th Symposium on Foundations of Computer Science},
  year={2024}
}

@article{psomas2022fair,
  title={Fair and efficient allocations without obvious manipulations},
  author={Psomas, Alexandros and Verma, Paritosh},
  journal={Advances in Neural Information Processing Systems},
  volume={35},
  pages={13342--13354},
  year={2022}
}

@inproceedings{gkatzelis2024getting,
  title={Getting more by knowing less: Bayesian incentive compatible mechanisms for fair division},
  author={Gkatzelis, Vasilis and Psomas, Alexandros and Tan, Xizhi and Verma, Paritosh},
  booktitle={Proceedings of the Thirty-Third International Joint Conference on Artificial Intelligence},
  pages={2807--2815},
  year={2024}
}

@inproceedings{lipton2004approximately,
  title={On approximately fair allocations of indivisible goods},
  author={Lipton, Richard J and Markakis, Evangelos and Mossel, Elchanan and Saberi, Amin},
  booktitle={Proceedings of the 5th ACM Conference on Electronic Commerce},
  pages={125--131},
  year={2004}
}

@article{schummer1996strategy,
  title={Strategy-proofness versus efficiency on restricted domains of exchange economies},
  author={Schummer, James},
  journal={Social Choice and Welfare},
  volume={14},
  number={1},
  pages={47--56},
  year={1996},
  publisher={Springer}
}

@article{garg2025designing,
  title={Designing Truthful Mechanisms for Asymptotic Fair Division},
  author={Garg, Jugal and Narayan, Vishnu V and Shen, Yuang Eric},
  journal={arXiv preprint arXiv:2512.10892},
  year={2025}
}

@inproceedings{amanatidis2017truthful,
  title={Truthful allocation mechanisms without payments: Characterization and implications on fairness},
  author={Amanatidis, Georgios and Birmpas, Georgios and Christodoulou, George and Markakis, Evangelos},
  booktitle={Proceedings of the 2017 ACM Conference on Economics and Computation},
  pages={545--562},
  year={2017}
}

@article{papai2000strategyproof,
  title={Strategyproof multiple assignment using quotas},
  author={P{\'a}pai, Szilvia},
  journal={Review of Economic Design},
  volume={5},
  number={1},
  pages={91--105},
  year={2000},
  publisher={Springer}
}

@book{shaked2007stochastic,
  title={Stochastic orders},
  author={Shaked, Moshe and Shanthikumar, J George},
  year={2007},
  publisher={Springer}
}

@article{dubins1961cut,
  title={How to cut a cake fairly},
  author={Dubins, Lester E and Spanier, Edwin H},
  journal={The American Mathematical Monthly},
  volume={68},
  number={1P1},
  pages={1--17},
  year={1961},
  publisher={Taylor \& Francis}
}

@article{stromquist1980cut,
  title={How to cut a cake fairly},
  author={Stromquist, Walter},
  journal={The American Mathematical Monthly},
  volume={87},
  number={8},
  pages={640--644},
  year={1980},
  publisher={Taylor \& Francis}
}

\appendix

\newpage

\section{Non-Identical Distributions}\label{app:non-identical-independent}

In this section, we extend our results to the setting where agents' valuations are independent but not necessarily identically distributed. We adopt the model of \citet{baienvy2022}. Instead of all agents drawing values from the same distribution $\mathcal{D}$, there are $n$ distributions $\mathcal{D}_1, \ldots, \mathcal{D}_n$, and each agent $i$'s value for item $j$ is drawn independently as  $v_{i,j} \sim \mathcal{D}_i$.

We assume (as \citeauthor{baienvy2022} did) that the distributions satisfy regularity conditions: The distributions admit probability density functions (PDFs) $f_1, \dots, f_n$ that satisfy
\begin{enumerate}
    \item Interval Support: For each $i$, there exists an interval $[a_i, b_i] \subseteq [0, 1]$ such that $f_i(x) > 0$ if and only if $x \in [a_i, b_i]$.
    \item Bounded PDFs: There exist constants $p, q > 0$ such that for all $i$ and all $x$ where $f_i(x) > 0$, we have $p \le f_i(x) \le q$.
\end{enumerate}

In this setting, maximizing social welfare is no longer an appropriate notion of efficiency, as agents' values maybe an arbitrarily different scales. Instead, we focus on \emph{Pareto optimality (PO)}. An allocation $(A_1, \ldots, A_n)$ is PO if the values $(u_1(A_1), \ldots, u_n(A_n))$ are not \emph{Pareto dominated}, i.e., there is no allocation $(A'_1, \ldots, A'_n)$ such that $u_i(A'_i) \ge u_i(A_i)$ for all $i$, with strict inequality for at least one agent. We say an allocation is $\alpha$-Pareto-Optimal ($\alpha$-PO) if $(u_1(A_1)/\alpha, \ldots, u_n(A_n)/\alpha)$ is not Pareto dominated.

\citeauthor{baienvy2022} show that under these regularity conditions, there exist weights $w_1, \ldots, w_n$ such that maximizing \emph{weighted} welfare, assigning item $j$ to an agent maximizing $w_i \cdot v_{i,j}$, yields an allocation that is PO and EF with high probability as $m$ grows large. These weights are chosen such that the probability that any specific agent $i$ has the highest weighted value for an item is exactly $1/n$. As in \citet{dickerson2014computational}, they in fact show that not only will the allocation be EF, but $u_i(A_i) - u_i(A_j) \in \Omega(m/n)$ with high probability.

We show that under the same conditions as \citet{baienvy2022}, a weighted version of the ranking algorithm is still BIC with high probability. Note that the algorithm will now depend on the distributions $\mathcal{D}$. 

\begin{algorithm}[htb!]
\caption{Weighted Ranking Algorithm for CDFs $F_1, \ldots, F_n$.}
\label{algo:ranking2}
\KwIn{$\{b_i\}$}
\BlankLine

$w_1, \ldots, w_n \gets$ Weights for $F_1, \ldots, F_n$ from \citet{baienvy2022}.

$\sigma_i \gets \text{Ranking of bids $b_i$, breaking ties arbitrarily} \quad  \forall i$\\
$A_i \gets \emptyset \quad \forall i$

\For{$j \gets 1$ \KwTo $m$}{
    $\hat{v}_{ij} \gets F^{-1}(1 - \rank(\sigma_i, j)/(m + 1))$\\
    $v \gets \max_{i} w_i \cdot \hat{v}_{ij}$ \\
    $i^* \gets \min_{i:\, =r} i$ \\
    $A_{i^*} \gets A_{i^*} \cup \{j\}$ \\
}
\end{algorithm}

\begin{theorem}\label{thm: bic non-iid theorem}
    For all interval-support and PDF-bounded distributions $\D_1, \ldots, \D_n$ with CDFs $F_1, \ldots, F_n$, \Cref{algo:ranking2} is BIC, EF, and $(1 - \varepsilon)$-PO with high probability. 
\end{theorem}
\begin{proof}   
    We first establish the BIC property. The logic follows the same as \Cref{thm: bic iid theorem}. From the perspective of agent $i$, the values (and thus the rankings) of all other agents are drawn independently across items. Consequently, the distribution of the "competing" weighted virtual values—$\max_{k \neq i} \{ w_k F_k^{-1}(q_{k,j}) \}$—is identical for every item $j$. This implies that for any rank position $r \in [m]$, there exists a fixed probability $p_r$ that agent $i$ wins the item they place at rank $r$. Crucially, because the CDF $F_i$ is monotone non-decreasing, placing an item at a higher rank (smaller $r$) results in a higher implied quantile $q_{i,j}$, a higher virtual value $\tilde{v}_{i,j}$, and strictly higher winning probability $p_r$. Thus, $p_1 \ge p_2 \ge \dots \ge p_m$. To maximize expected utility $\sum_{j} v_{i,j} \cdot \mathbb{P}(\text{win } j)$, agent $i$ must assign higher probabilities to higher-valued items. Truthful reporting (ranking by value) is therefore a dominant strategy in the induced game, establishing BIC.

    Next, we address fairness and efficiency. We show that the allocation produced by \Cref{algo:ranking2} differs from the ideal weighted-welfare maximization on only a vanishing fraction of items. The proof is again similar to \Cref{thm: bic iid theorem}, except we need to handle the different values with slightly more care.

    First, we claim it is without loss of generality to assume all weights $w_i = 1$. Indeed,
    instead of working directly with these distributions, we will consider the \emph{weighted} distributions $w_1 \mathcal{D}_1, \ldots, w_n \mathcal{D}_n$. Note that these will all be  $p'$-$q'$-PDF bounded for $p' = p/\max_i w_i$ and $q' q/\min_i w_i$. Our proof will hold for these weighted distributions, and both of these lead to indistinguishable algorithms. 

    Let $\hat{F}_i$ denote the empirical CDF of agent $i$'s values. By the DKW inequality, with probability $1 - O(1/m^2)$, we have $\sup_x | \hat{F}_i(x) - F_i(x) | \le \sqrt{(\ln m)/m}$. Because the PDFs are bounded above by $q'$, the error in the value space is also bounded. Specifically, for any quantile $x$, the difference between the true value and the value implied by the empirical rank is bounded by $q' \cdot O(\sqrt{(\ln m)/m})$.

    Consider an item $j$. \Cref{algo:ranking2} assigns $j$ to the agent with the highest virtual value. An error occurs only if the agent with the highest true value, say $i^*$, loses to some agent $k$ because of the ranking error. This requires the true values $v_{i^*,j}$ and $v_{k,j}$ to be within $O(\sqrt{(\ln m)/m})$ of each other. Since the difference of two independent random variables with bounded PDFs also has a bounded PDF, the probability that $|v_{i^*,j} - v_{k,j}| \le \epsilon$ is $O(\epsilon)$. Substituting our error bound, the probability that a specific item is allocated ``incorrectly" (compared to exact weighted-welfare maximization) is $O(\sqrt{(\ln m)/m})$. Summing over all items and applying Hoeffding’s inequality, the total number of mismatched items is $O(\sqrt{m \ln m})$ with high probability. As established by \citet{dickerson2014computational} and \citet{baienvy2022}, the ideal weighted allocation is robust to adversarial changes of $o(m)$ items. Thus, \Cref{algo:ranking2} retains the properties of Envy-Freeness and $(1-\epsilon)$-Pareto Optimality with high probability.
\end{proof}

\section{A cake-cutting approach to correlated values}\label{app:cake-cutting}

In this section, we consider the setting where item values are drawn from a correlated distribution $\mathcal{D}$. That is, the valuation vectors $(v_{i,1}, \dots, v_{i,n}) \sim \mathcal{D}$ are correlated across agents, though we maintain the assumption that values are independent and identically distributed across items. This matches the model of, e.g., \citet{benade2024fair}.

We focus on the case of $n=2$ agents. We show that under a mild assumption on the correlation structure of $\mathcal{D}$—specifically, that the agents' quantile rankings are not perfectly dependent—there exists a Pick-$r$ mechanism that yields an allocation that is Envy-Free with high probability.

\textbf{Correlation Assumption.} We assume that for any agent $i$ and any quantile threshold $q \in (0,1)$, the set of items agent $i$ ranks above the $q$-th quantile is not identical to the set of items agent $j$ ranks above the $q$-th quantile with probability 1. Formally, if $F_i$ is the marginal CDF for agent $i$, we assume that for any $q$ where the upper tail has non-zero probability, $\mathbb{P}(F_j(v_j) \ge q \mid F_i(v_i) \ge q) < 1$.

\paragraph{Two agents.} Without loss of generality, normalize the distribution $\mathcal{D}$ such that $\mathbb{E}[v_{i,j}] = 1$ for all agents $i$. We define a critical quantile threshold $\tau_i \in [0,1]$ for each agent $i$. This threshold represents the quantile cutoff required for the agent to achieve exactly half of their total expected value. Formally, $\tau_i$ satisfies:$$\mathbb{E}_{v \sim F_i}[v \cdot \mathbb{I}(F_i(v) \ge \tau_i)] = \frac{1}{2}\mathbb{E}[v] = \frac{1}{2}.$$Let $i$ be an agent such that $\tau_i \ge \tau_j$ for $j \ne i$. We run the Pick-$r$ mechanism with agent $i$ picking first, setting $r$ to be the fraction of items with mass $(1-\tau_i) + \varepsilon$ for sufficiently small $\varepsilon$. The key observation is that agent $i$ will value the items they receive at strictly more than half of their total value (and will therefore not envy $j$). Furthermore, because $\tau_i \ge \tau_j$ and valuations are not perfectly correlated, agent $j$'s expected value for agent $i$'s bundle is strictly less than $1/2$. Hence, for sufficiently large $m$, concentration inequalities ensure that these expected preferences translate to an EF allocation with high probability. Furthermore, for two agents, envy-freeness implies a $1/2$-approximation to maximum welfare (as all agents are receiving at least $1/2$ of their maximum possible utility).\paragraph{Extensions to $n > 2$ agents and cake cutting.}For $n > 2$ agents, the problem becomes significantly more challenging. Our procedure for $n=2$ can be viewed as a discrete analogue of the Dubins-Spanier Moving-Knife procedure~\cite{dubins1961cut}, where a knife moves across ``quantile space'' rather than across the physical cake. Under mild correlation assumptions, this approach can be generalized to guarantee a proportional allocation using a DSIC mechanism for any number of agents. In standard cake-cutting, \emph{contiguous} envy-free allocations always exist for any $n$~\cite{stromquist1980cut}. It remains an open question whether this result extends to \emph{quantile-contiguous} allocations, where there is an ordering of the agents $\sigma(1), \ldots, \sigma(n)$, where each agent $\sigma(i)$ receives the portion of the cake with highest density among the remaining pieces of $\sigma(i), \ldots, \sigma(n)$. The existence of such an allocation would ensure that there is a DSIC algorithm for $n$ agents that is EF with high probability, even with correlated valuations. Furthermore, it is unclear how to combine these properties with efficiency.

\section{Proof of Lemma \ref{lemma:high_welfare_conf_on_h}}\label{sec:high_welfare_conf_on_h}

\begin{proof}[Proof of~\Cref{lemma:high_welfare_conf_on_h}]
We will proceed by casework on $h$.

For $h=1$, a single high-value draw exists for one agent-item pair, $(i, j)$. To realize this value in the final allocation, agent $i$ must receive item $j$. We will show that with probability at least $1/4$ (over the choice of $(i, j)$ and other item values), agent $i$ does not receive item $j$.

For intuition, we first show why this is the case for an equal-sized trade, where $\ell = k$. Suppose without loss of generality that $i = 1$. With probability $1/2$, $j$ is not endowed to $i$. From agent 2's perspective, all items are worth either $0$ or $1$, and due to the equally sized bundle, they will prefer agent $1$'s bundle only with probability $1/2$. Thus, with probability $1/4$, the high-value item is not endowed to $1$ and agent $2$ is not willing to trade, which means that with probability $1/4$ agent $i$ does not receive item $j$, which proves Lemma \ref{lemma:high_welfare_conf_on_h} for the case when $\ell = k$ and $i = 1$.

Next, we extend this to unbalanced trades with $\ell \ne k$ but still with $i = 1$. 
Let $\tau$ be the probability that, if we sample values for bundles of size $\ell$ and $k$ from $\text{Bern}(q) + \text{Unif}[0, \delta]$, the $\ell$-bundle is more valuable than the $k$-bundle. Since ties occur with probability $0$, this means that the $k$-bundle is more valuable with probability $1 - \tau$. Without loss of generality, assume $\ell \ge k$, which implies that $\tau \ge 1/2$.

Note that if agent 1 has the high value, then agent 2 will be willing to trade with probability $1 - \tau$; conversely, if agent 2 has the high value, then agent 1 will be willing to trade with probability $\tau$. There are two events of interest: (i) Agent 2 has high value, it is in agent 1's bundle, and agent 1 is not willing to trade and (ii) Agent 1 has the high value, it is in agent 2's bundle, and agent 2 is not willing to trade. The first occurs with probability $\ell/(2s) \cdot \tau$ and the second occurs with probability $k/(2s) \cdot (1 - \tau)$. These are disjoint events, so the probability that at least one occurs is
\[
    \ell/(2s) \cdot \tau + k/(2s) \cdot (1 - \tau)  = \frac{\ell \tau + k (1-\tau)}{2s} \ge \frac{\ell/2 + k/2}{2s} = 1/4,
\]
where the inequality holds because $\tau \ge 1/2$ and $\ell \ge k$.

Thus, the single high-value item $i$ is not received by agent $j$ with probability at least $1/4$. Therefore, the expected welfare contribution of this single item is at most
$(1-1/4)\left(\frac{1-q}{p}+\delta\right) \le \frac{3}{4} \cdot h \cdot \frac{1-q}{p} + \delta h$.

We will now show the $h > 1$ case. Fix any $h > 1$. Let $r$ be the number of high-value items in endowed bundles, and $h - r$ be the remaining number in non-endowed bundles. The welfare contribution from high-value items in this case is bounded by 
\begin{equation}\label{eq:max_r_h}
\max(r, h-r) \cdot \frac{1 - q}{p} + \delta h= \left(\frac{h}{2} + |r - \frac{h}{2}| \right) \frac{1-q}{p} + \delta h.
\end{equation}
Therefore, to bound the expected welfare contribution from high-value items in this case,  we will bound $\E[|r - \frac{h}{2}|] \le \frac{h}{4}$. 

The distribution of $r$ is a hypergeometric with parameters $N = 2s, K = s, n = h$. Indeed, we are sampling $h$ pairs $(i, j)$ without replacement, and exactly half are in endowed bundles. However, for our purposes, it turns out that this distribution is convex dominated by a $r' \sim \text{Bin}(h, 1/2)$ distribution. Since $|r - \frac{h}{2}|$ is convex in $r$, $\E_{r \sim \text{HypGeo}(2s, s, h)}[|r - \frac{h}{2}|] \le \E_{r' \sim  \text{Bin}(h, 1/2)}[|r' - \frac{h}{2}|]$ (see Theorem 3.A.39 of \cite{shaked2007stochastic}). We will instead upper bound the binomial version.

For any $h \ge 4$, by Jensen's inequality:
\[\E[|r'-h/2|] \le \sqrt{\E[(r'-h/2)^2]} = \sqrt{\text{Var}(r')} = \sqrt{h/4} = \sqrt{h}/2 \le h/4.\]

For $h=2$: $\E[|r'-1|] = \frac{1}{4}|0-1| + \frac{1}{2}|1-1| + \frac{1}{4}|2-1| = 2/4$.

For $h=3$: $\E[|r'-3/2|] = \frac{1}{8}|0-3/2| + \frac{3}{8}|1-3/2| + \frac{3}{8}|2 - 3/2| + \frac{1}{8}|3 - 3/2| = 3/4$. 

Thus, for all $h > 1$ we have that
\[
\E[|r - h/2|] \le \E[|r' - h/2|] \le h/4.
\] 
Therefore, by Equation \eqref{eq:max_r_h}, we can conclude that the expected welfare contributed for high-value items is bounded by
\begin{align*}
    \E\left[\left(\frac{h}{2} + |r - \frac{h}{2}| \right) \frac{1-q}{p} + \delta h\right] &\le \left(\frac{h}{2} + \E[|r - \frac{h}{2}|] \right) \frac{1-q}{p} + \delta h \\
    &\le \left(\frac{h}{2} + \E[|r' - \frac{h}{2}|] \right) \frac{1-q}{p} +\delta h\\
    &\le \left(\frac{h}{2} + \frac{h}{4} \right) \frac{1-q}{p} +\delta h \\
    &= \frac{3h}{4} \frac{1-q}{p} + \delta h
\end{align*}
as desired to complete the proof of~\Cref{lemma:high_welfare_conf_on_h}.
\end{proof}

\section{Proof of Lemma \ref{lem:equiv}}\label{sec:equiv}

\begin{proof}[Proof of \Cref{lem:equiv}]
First, observe that in QT-$\mathbf{s}$, the probability that agent $i$ receives a specific item is exactly $r_i$. This follows by strong induction: the probability the item reaches agent $i$ is $1 - \sum_{k=1}^{i-1} r_k$, and conditioned on reaching $i$, it is taken with probability $s_i$. The product of these terms is $r_i$.

We begin with (1). Let $(A_1, \ldots, A_n)$ denote the allocation of Pick-$\mathbf{r}$ and $(A'_1, \ldots, A'_n)$ denote the allocation of QT-$\mathbf{s}$. 

By the above argument, we have that $|A'_i|$ follows a binomial distribution with expectation $\E[|A'_i|] = r_i m$. 
Let $C = \sqrt{m \log(m)} + 1 \in O(\sqrt{m\log(m)})$. Applying Hoeffding's inequality, $\Pr[||A'_i| - r_i \cdot m| > C - 1] \le 2 \exp(-2(\sqrt{m \log m})^2 / m) = 2/m^2$. Furthermore, $||A_i| - r_i \cdot m| \le 1$. Combining these, under the Hoeffding event ($||A'_i| - r_i \cdot m| > C - 1$), we have that $||A'_i| - |A_i|| \le C$. By a union bound over $n$ agents, this holds for all agents with probability at least $1 - O(n/m^2)$. We will condition on the event that $||A_i| - |A'_i|| \le C$, for all $i$, for the remainder of the proof.

Fix arbitrary quantiles $q_{i,j}$ such that the event holds, and condition on the probability $1$ event that they are all distinct.

Fix an agent $i$. Let $U_i$ and $U'_i$ be the set of available items for agent $i$ in Pick-$\mathbf{r}$ and QT-$\mathbf{s}$ respectively (where $U_i = \mathcal{M} \setminus \bigcup_{k = 1}^{i - 1} A_i$ and $U'_i = \mathcal{M} \setminus \bigcup_{k = 1}^{i - 1} A'_i$). Let $d_i = |U_i \Delta U'_i|$ be the number of items available in one mechanism but not the other. We define $U^*_i = U_i \cap U'_i$ as the set of items available in both mechanisms.

Both mechanisms are greedy with respect to the same values (quantiles). Therefore, when choosing from the common set $U^*_i$:
\begin{itemize}
    \item Pick-$\mathbf{r}$ selects the top $y_i =: |A_i \cap U^*_i|$ items from $U^*_i$.
    \item QT-$\mathbf{s}$ selects the top $y_i' =: |A'_i \cap U^*_i|$ items from $U^*_i$.
\end{itemize}
Since both select the best items from the same set, the selected subsets are nested. Specifically, if $y_i \le y'_i$, then $(A_i \cap U^*_i) \subseteq (A'_i \cap U^*_i)$, and vice versa. The number of disagreements on these common items is exactly the difference in count: $$|(A_i \cap U^*_i) \Delta (A'_i \cap U^*_i)| = |y_i - y'_i|.$$

We now bound the total allocation disagreement $|A_i \Delta A'_i|$. Let $x_i = |A_i \setminus U^*_i|$ be the items agent $i$ took from $U_i \setminus U'_i$ (items only available in Pick-$\mathbf{r}$). Let $x'_i = |A'_i \setminus U^*_i|$ be the items agent $i$ took from $U'_i \setminus U_i$ (items only available in QT-$\mathbf{s}$). Note that $x_i \le |U_i \setminus U'_i|$ and $x'_i \le |U'_i \setminus U_i|$, so $x_i + x'_i \le d_i$.

The total number of items where allocations $A_i$ and $A'_i$ disagree is the sum of disagreements from the disjoint sets and the common set:$$|A_i \Delta A'_i| = x_i + x'_i + |y_i - y'_i|.$$

We bound $|y_i - y'_i|$ using the bundle size constraint. We know $|A_i| = x_i + y_i$ and $|A'_i| = x'_i + y'_i$. From our conditioning, we established $| |A_i| - |A'_i| | \le C$. Substituting the sums:$$|(x_i + y_i) - (x'_i + y'_i)| \le C \implies |y_i - y'_i| \le C + |x_i - x'_i|.$$ Substituting this back into the disagreement equation:$$|A_i \Delta A'_i| = x_i + x'_i + |y_i - y'_i| \le x_i + x'_i + C + |x_i - x'_i|.$$ Since $x_i, x'_i \ge 0$, we have $|x_i - x'_i| \le x_i + x'_i$. Thus:$$|A_i \Delta A'_i| \le 2(x_i + x'_i) + C \le 2d_i + C.$$

Furthermore, from these values, we can derive $d_{i+1}$ (the disagreement passed to the next agent), which will allow use to define a recurrence. The set of available items updates as $U_{i+1} = U_i \setminus A_i$ and $U'_{i+1} = U'_i \setminus A'_i$. An item is in the symmetric difference $U_{i+1} \Delta U'_{i+1}$ if it was already a disagreement in $d_i$ and was \emph{not} picked, or if it was in the common set $U^*_i$ and was picked by exactly one mechanism (creating a new disagreement).$$d_{i+1} = \underbrace{(d_i - (x_i + x'_i))}_{\text{Unpicked input disagreements}} + \underbrace{|y_i - y'_i|}_{\text{New disagreements from } U^*_i}.$$

Substituting $|y_i - y'_i| \le C + |x_i - x'_i|$:$$d_{i+1} \le d_i - (x_i + x'_i) + C + |x_i - x'_i|.$$Using $|x_i - x'_i| \le x_i + x'_i$, the terms cancel to satisfy:$$d_{i+1} \le d_i + C.$$Since $d_1 = 0$ (both start with all items), it follows by induction that $d_i \le (i-1)C$.

Plugging this in above, we have that $|A_i \Delta A'_i| \le (2i - 1) C$. Summing over all $n$ agents:$$\sum_{i=1}^n |A_i \Delta A'_i| \le \sum_{i=1}^n (2i - 1)C = n^2 C.$$Substituting $C = O(\sqrt{m \log m})$, the total number of disagreement items is $O(n^2 \sqrt{m \log m})$.

We now turn to (2). Consider the expected welfare of QT-$\mathbf{s}$. Recall that we showed that each agent $i$ receives each item with probability $r_i$. These allocation decisions are made independently. Conditioned on item $j$ being allocated to agent $i$, the value $v_{ij}$ is guaranteed to be at least $F^{-1}(1-s_i)$. Thus, the expected welfare contribution of a single item is: $$w_{\text{avg}} = \sum_{i=1}^n r_i \cdot \E[X \mid X \ge F^{-1}(\tau_i)].$$ By linearity of expectation, the total expected welfare of $\text{QT-}\mathbf{s}$ is $ m \cdot w_{\text{avg}} = \alpha \cdot m \cdot \E[\max_i X_i]$ (by using the equality $\alpha = (\sum_{i=1}^n r_i \cdot \E[X \mid X \ge F^{-1}(\tau_i)]) / \E[\max_i X_i]$).

Now, consider Pick-$\mathbf{r}$. Conditioned on the event that the allocations differ by at most $O(n^2 \sqrt{m \log m})$ items, the difference in welfare can be at most $O(n^2 \sqrt{m \log m})$. Furthermore, in the remaining case (that occurs with $O(n/m^2)$ probability), the welfare difference is at most $m$. Hence, the contribution to the expected welfare is at most $m \cdot O(n/ m^2) = O(n/m)$. 

Dividing by the expected maximum welfare of $m \cdot \E[\max_i X_i]$, we get an approximation of
\[
    \frac{\alpha \cdot m \cdot \E[\max_i X_i] -O(n^2 \sqrt{m \log m}) - O(n/m)}{m \cdot \E[\max_i X_i]} = \alpha - o(1).
\]

Finally, we prove (3). 
    We must show that for any pair of agents $i, j$, the probability that agent $i$ envies agent $j$ approaches 0 as $m \to \infty$.

    We first analyze the valuations in QT-$\mathbf{s}$. Since valuations are independent across items, we calculate the expected value agent $i$ assigns to a single random item allocated to $i$ versus one allocated to $j$.
    
    \begin{itemize}
        \item \emph{Value for Own Bundle ($A_i$):}
        In $\text{QT-}\mathbf{s}$, agent $i$ receives an item with probability $r_i$. Conditioned on receipt, the item's value is at least $\tau_i$. Thus, the expected total value is:
        $$ \E[v_i(A_i)] = m \cdot r_i \cdot \E[X \mid X \ge \tau_i]. $$
        
        \item \emph{Value for Other Bundle ($A_j$) when $i < j$:}
        Since $i$ precedes $j$ in the priority sequence, $j$ can only receive an item if $i$ has already rejected it (implying $v_{i,k} < \tau_i$) and $j$ subsequently accepts it. Since valuations are independent across agents, the event that $j$ accepts the item provides no information about $i$'s value other than the fact that $i$ rejected it. Therefore, $i$'s expected value for an item in $j$'s bundle is conditioned only on $X < \tau_i$:
        $$ \E[v_i(A_j)] = m \cdot r_j \cdot \E[X \mid X < \tau_i]. $$
        
        \item \emph{Value for Other Bundle ($A_j$) when $i > j$:}
        Here, $j$ considers items before $i$. If $j$ selects an item, it is removed before $i$ sees it. Since $i$'s valuation is independent of $j$'s, the event that $j$ picked the item is independent of $i$'s specific value for it. Thus, $i$'s expected value for items in $A_j$ follows the unconditional distribution:
        $$ \E[v_i(A_j)] = m \cdot r_j \cdot \E[X]. $$
    \end{itemize}

    The conditions given in the lemma statement correspond exactly to the requirement that agent $i$ strictly prefers their expected bundle to agent $j$'s.
    Let $\Delta_{ij} = \E[v_i(A_i)] - \E[v_i(A_j)]$. In both cases, the strict inequality implies that $\Delta_{ij} = c \cdot m$ for some constant $c > 0$.

    In $\text{QT-}\mathbf{s}$, both $v_i(A_i)$ and $v_i(A_j)$ are sums of independent random variables bounded in $[0, 1]$. By Hoeffding's inequality, the realized difference $v_i(A_i) - v_i(A_j)$ concentrates around its expectation $\Delta_{ij}$ with deviations bounded by $O(\sqrt{m \log m})$ with high probability.
    
    Finally, we transfer this result to the $\text{Pick-}\mathbf{r}$ mechanism. From Part (1), we know that with high probability, the allocations of $\text{Pick-}\mathbf{r}$ and $\text{QT-}\mathbf{s}$ differ on at most $K = O(n^2 \sqrt{m \log m})$ items. Since values are bounded by 1, this difference affects the total valuation by at most $K$.
    Thus, for the $\text{Pick-}\mathbf{r}$ mechanism:
    $$ v_i(A_i) - v_i(A_j) \ge \underbrace{c \cdot m}_{\text{Expected Gap}} - \underbrace{O(\sqrt{m \log m})}_{\text{Concentration Noise}} - \underbrace{O(n^2 \sqrt{m \log m})}_{\text{Coupling Noise}}. $$
    For sufficiently large $m$, the linear term $c \cdot m$ dominates the sub-linear noise terms. Therefore, $v_i(A_i) > v_i(A_j)$ holds with probability approaching 1.
\end{proof}

\section{Proof of Theorem \ref{thm: main positive result for two agents}}\label{sec: main positive result for two agents}

\begin{proof}[Proof of \Cref{thm: main positive result for two agents}]

    For the rest of this proof, let $r = \frac{2-\sqrt{2}}{2}$. Fix any distribution $\mc{D}$. 

    By Part 3 of \Cref{lem:equiv}, the Pick-$r$ mechanism is envy-free with high probability if  
    \[
     r \cdot \E_{X \sim \mathcal{D}}[X \mid F_{\mathcal{D}}(X) \ge 1-r] > (1-r) \cdot \E_{X \sim \mathcal{D}}[X \mid F_{\mathcal{D}}(X) < 1-r],
    \]   
    Note that because $r < 0.5$, the other condition from Part 3 of \Cref{lem:equiv} holds immediately. 
    
    By  \Cref{thm:pos-welf-2}, if the above equation holds then we have that Pick-$(\tfrac{2-\sqrt{2}}{2})$ is both envy-free with high probability and a $\frac{2+\sqrt{2}}{4} - o(1)$ approximation to welfare, proving the desired result in this case.

    Now suppose that
    \[
     r\E_{X \sim \mathcal{D}}[X \mid F_{\mathcal{D}}(X) \ge 1-r] \le (1-r)\E_{X \sim \mathcal{D}}[X \mid F_{\mathcal{D}}(X) < 1-r].
    \]
    We will show (later in the proof) that the above equation implies that Pick-$  (\tfrac{1}{2}-\tfrac{1}{m^{0.25}})$ must be a $7/8 - o(1)\ge \frac{2+\sqrt{2}}{4} - o(1)$ approximation to optimal social welfare. Furthermore, for sufficiently large $m$ and any distribution $\mathcal{D}$ with non-$0$ variance, the Pick-$(\tfrac{1}{2}-\tfrac{1}{m^{0.25}})$ mechanism is envy-free with high probability by Hoeffding's Inequality. This is because agent 2 is allocated an extra $2m^{0.75}  \pm o(1)$ items and therefore prefers their bundle with high probability by Hoeffding's inequality. Agent 1 is getting $2m^{0.75} \pm o(1)$ fewer items than agent 2, but we can still show that they are envy-free with high probability as follows. The non-zero variance of $\mathcal{D}$ means that there must be some constant $\varepsilon > 0$ such that in expectation, agent 1 values their $(\tfrac{1}{2}-\tfrac{1}{m^{0.25}})m$ chosen items $\varepsilon m$ more in expectation than agent 1 values the items they leave for agent 2. Again by Hoeffding's Inequality, with high probability, this means that agent 1 strictly prefers their bundle even if agent  2 receives $2m^{0.75}$ more items (because $\varepsilon m > O(m^{0.75})$ for sufficiently large $m$). This means that with high probability, both agents do not envy each other, so Pick-$(\tfrac{1}{2}-\tfrac{1}{m^{0.25}})$ is EF with high probability, thus completing the proof of the main theorem in this case as well.
    
    The rest of the proof will be spent showing that 
    \[
     r\E_{X \sim \mathcal{D}}[X \mid F_{\mathcal{D}}(X) \ge 1-r] \le (1-r)\E_{X \sim \mathcal{D}}[X \mid F_{\mathcal{D}}(X) < 1-r].
    \]
    implies that Pick-$  (\tfrac{1}{2}-\tfrac{1}{m^{0.25}})$ must be  a $7/8 - o(1)$ approximation to optimal social welfare
    
    Rewriting the expectations as integrals gives
    \[
        r\int_0^{\infty} \Pr(X > s \mid F_{\mathcal{D}}(X) \ge 1-r)ds \le (1-r) \int_0^{\infty} \Pr( X > s \mid F_{\mathcal{D}}(X) < 1-r)ds,
    \]
    which simplifies to
    \[
    \int_0^{\infty} \min\left(r, 1-F_{\mathcal{D}}(s)\right) ds  \le \int_0^{\infty} \max\left(0,1-r-F_{\mathcal{D}}(s)\right) ds.
    \]
    Finally, subtracting the RHS gives
    \begin{equation}\label{eq:pick_r_not_EF}
    \int_0^{\infty} \left(r - |F_{\mathcal{D}}(s) - (1-r)| \right) ds  \le 0.
    \end{equation}
    Define $A$ as the LHS of this equation, i.e.

    \[
        A :=  \int_0^{\infty} \left(r - |F_{\mathcal{D}}(s) - (1-r)| \right) ds .
    \]
    Let $Q_{ALG}$ be the Pick-$(\tfrac{1}{2}-\tfrac{1}{m^{0.25}})$ mechanism and $Q_{OPT}$ be the welfare-maximizing mechanism. To prove that Pick-$(\tfrac{1}{2}-\tfrac{1}{m^{0.25}})$ is a $7/8-o(1)$ approximation to optimal social welfare, it suffices to show that 
        \[
        \int_{0}^{\infty} \Pr(Q_{ALG} > F_{\mathcal{D}}(s))ds \ge \left( \frac{7}{8}- o(1) \right)   \int_{0}^{\infty} \Pr(Q_{OPT} > F_{\mathcal{D}}(s))ds.
    \]
    Because $A \le 0$, to show the above equation it is sufficient to show that 
        \[
        \int_{0}^{\infty} \Pr(Q_{ALG} > F_{\mathcal{D}}(s)) +  \frac{A}{4}\ge \left( \frac{7}{8} - o(1) \right)   \int_{0}^{\infty} \Pr(Q_{OPT} > F_{\mathcal{D}}(s))ds.
    \]
    Also, as in \Cref{thm:pos-welf-2}, to show the above equation, it is sufficient to show that for all $t \in [0,1]$,    
    \begin{equation}\label{eq:need_to_show_ef}
        \Pr(Q_{ALG} > t) + \frac{r - |t - (1-r)|}{4} \ge \left( \frac{7}{8}- o(1) \right) (1-t^2).
    \end{equation}
    By the same logic as in \Cref{thm:pos-welf-2} , 
    \[
    \Pr(Q_{ALG} > t)  = 
    \begin{cases}
    1/2 + (1-t)/2 \pm o(1) & \text{if } t < 1/2, \\
   1 - t + (1-t)/2 \pm o(1) & \text{if } t \geq 1/2.
    \end{cases}
    \]
    Therefore, 
    \begin{equation}\label{eq:piecewise_ef}
       \Pr(Q_{ALG} > t) + \frac{r - |t - (1-r)|}{4}  = \begin{cases}
            \dfrac{3+2r-t}{4} \pm o(1), & 0 \le t < \dfrac{1}{2},\\[6pt]
            \dfrac{5+2r-5t}{4}\pm o(1), & \dfrac{1}{2} \le t < 1-r,\\[6pt]
            \dfrac{7-7t}{4}\pm o(1), & 1-r \le t \le 1.
\end{cases}
    \end{equation}
    Define 
    \[
    S(t) = \begin{cases}
            \dfrac{3+2r-t}{4} , & 0 \le t < \dfrac{1}{2},\\[6pt]
            \dfrac{5+2r-5t}{4}, & \dfrac{1}{2} \le t < 1-r,\\[6pt]
            \dfrac{7-7t}{4}, & 1-r \le t \le 1.
\end{cases}
    \]
    We will spend the rest of the proof showing that $S(t) \ge \frac{7}{8}(1-t^2)$ for all $t \in [0,1]$, which together with Equation \eqref{eq:piecewise_ef} will imply Equation \eqref{eq:need_to_show_ef}, thereby completing the proof. Define $D(T) = S(T) - \frac{7}{8}(1-t^2)$.
    
On each piece, $S(t)$ is affine, so $D(t)$ is a convex quadratic. This means to show that $D(T)$ is non-negative, it suffices to check its minimum on each piece.

\textbf{1) $0\le t<\tfrac12$.}
\[
D(t)=\frac{7}{8}t^2-\frac14\,t+\left(\frac{5-\sqrt{2}}{4}-\frac78\right),\qquad
D'(t)=\frac{7}{4}t-\frac14.
\]
The critical point is at $t=\tfrac{1}{7}\in[0,\tfrac12)$, hence
\[
D_{\min}=D\!\left(\tfrac{1}{7}\right)
=\frac{10-7\sqrt{2}}{28}>0.
\]

\textbf{2) $\tfrac12\le t<1-r$.}
\[
D(t)=\frac{7}{8}t^2-\frac{5}{4}\,t+\left(\frac{5+2r}{4}-\frac78\right),\qquad
D'(t)=\frac{7}{4}t-\frac{5}{4}.
\]
The root of $D'(t)$ is $t=\tfrac{5}{7}>1-r=\tfrac{\sqrt{2}}{2}$, so $D'(t)<0$ on this interval and the minimum occurs at the right endpoint:
\[
D_{\min}=D(1-r)=\frac{7}{8}(r)^2
=\frac{7}{8}\!\left(\frac{3}{2}-\sqrt{2}\right)>0.
\]

\textbf{3) $1-r\le t\le 1$.}
\[
D(t)=\frac{7}{4}(1-t)-\frac{7}{8}(1-t^2)
=\frac{7}{8}(t-1)^2\ge 0,
\]
with equality only at $t=1$.

\medskip
Combining the three cases above, we have shown that $D(t) \ge 0$ for all $t \in [0,1]$, completing the proof.   
\end{proof}

\end{document}